\documentclass[10pt]{article}
\usepackage{bbm}
\usepackage{amsmath,amsthm,amssymb}
\usepackage{tabularx}
\usepackage{enumitem}
\usepackage{mathrsfs}
\usepackage{caption}
\usepackage{graphicx} 
\usepackage{authblk}
\usepackage[hypertexnames=false,colorlinks=true,urlcolor=blue,linkcolor=blue,citecolor=blue]{hyperref}
\usepackage[numbers,comma,square,sort&compress]{natbib}
\usepackage[letterpaper,text={6.5in,9in},centering]{geometry}

\graphicspath{{Figures/}}

\makeatletter\@addtoreset{equation}{section}\makeatother
\renewcommand{\theequation}{\arabic{section}.\arabic{equation}}

\makeatletter
\newtheorem*{rep@theorem}{\rep@title}
\newcommand{\newreptheorem}[2]{%
\newenvironment{rep#1}[1]{%
 \def\rep@title{#2 \ref*{##1}}%
 \begin{rep@theorem}}%
 {\end{rep@theorem}}}
\makeatother

\newtheorem{thm}{Theorem}[section] 
\newtheorem{lem}[thm]{Lemma}  
\newtheorem{cor}[thm]{Corollary} 
  
\newtheorem{hyp}{Hypothesis}
\newtheorem*{hyp*}{Hypotheses}

\newtheorem{rmk}[thm]{Remark}
\newtheorem{rslt}{Result}

\newreptheorem{lem}{Lemma}
\newreptheorem{thm}{Theorem}

 {\begin{trivlist}\item[]\textbf{Acknowledgments.}}{\end{trivlist}}

 {\begin{trivlist}\item[]\textbf{Data availability statement.}}{\end{trivlist}}

\begin{document}

\title{Global bifurcation of localised 2D patterns emerging from spatial heterogeneity}
\author[1,2]{Dan J. Hill}
\author[3]{David J.B. Lloyd}
\author[3]{Matthew R. Turner}
\affil[1]{\small Mathematical Institute, University of Oxford, Oxford, OX2 6GG, UK}
\affil[2]{\small Fachrichtung Mathematik, Universit\"at des Saarlandes, Postfach 151150, 66041 Saarbr\"ucken, Germany}
\affil[3]{\small School of Mathematics and Physics, University of Surrey, Guildford, GU2 7XH, UK}

\maketitle
\begin{abstract}

\noindent We present a general approach to prove the existence, both locally and globally in amplitude, of fully localised multi-dimensional patterns in partial differential equations containing a compact spatial heterogeneity. While one-dimensional localised patterns induced by spatial heterogeneities have been well-studied, proving the existence of fully localised patterns emerging from a Turing instability in higher dimensions remains a key open problem in pattern formation. In order to demonstrate the approach, we consider the two-dimensional Swift--Hohenberg equation, whose linear bifurcation parameter is perturbed by a radially-symmetric potential function. In this case, the trivial state is unstable in a compact neighbourhood of the origin and linearly stable outside. We prove the existence of local bifurcation branches of fully localised patterns, characterise their stability and bifurcation structure, and then rigorously continue solutions to large amplitude via analytic global bifurcation theory. Notably, the primary bifurcating branch in the Swift--Hohenberg equation alternates between an axisymmetric spot and a non-axisymmetric `dipole' pattern, depending on the width of the spatial heterogeneity.


\end{abstract}

\section{Introduction}

In this paper we prove the global existence of fully localised two-dimensional patterns close to a pattern-forming or Turing instability. The key breakthrough is to use a compact region of spatial heterogeneity to induce a pattern-forming bifurcation away from the Turing point, where the linearised operator remains Fredholm. Although there are many proofs---particularly in the area of water waves---of global bifurcations of waves and patterns, the extension to fully localised multi-dimensional patterns remains elusive. As such, this work provides a significant  novel direction for future research.

The emergence of patterns with spatially localised extent are important for understanding phenomena in a wide range of applications from the formation of tornadoes~\cite{Navarro2015vortices-instab,Castano2018vortices-chaos}, patches and fairy circles in dryland vegetation~\cite{Byrnes2023Spots,Jaibi2020localised,Hill2024Predict,Mau2013}, and spikes on the surface of a magnetic fluid~\cite{Lloyd2015ferro-snaking,Richter2005ferro-spikes}; see~\cite{Bramburger2025review} for a review of localised pattern formation. The key bifurcation of these localised patterns from quiescence is due to a pattern-forming instability~\cite{Bramburger2025review}. While there is extensive theory for the emergence of one-dimensional localised patterns, to date there remains no existence proof for fully localised 2D patterns emerging from a Turing instability beyond the axisymmetric case~\cite{lloyd2009localized,mccalla2013spots,mcquighan2014oscillons,Hill2024n-dimension}. Fundamentally, this is due to the fact that the kernel of the linear operator at a 1D pattern-forming instability is finite dimensional but in 2D it is infinite-dimensional. We note that there have been recent advances in approximating these patterns~\cite{Hill2023dihedral-spot,Hill2024dihedral-ring,Hill2024amplitude}, as well as in proving their existence at certain parameter regimes~\cite{Aalst2025,Cadiot2025,Boissoniere2022} via rigorous numerics.

The present work is inspired by the use of spatial heterogeneity to induce the emergence of localised patterns in experiments. Two examples are the localised temperature hotspots that can lead to the nucleation of tornadoes in numerical simulations~\cite{Navarro2015vortices-instab,Castano2018vortices-chaos}, and the local amplification of a magnetic field inducing axisymmetric spikes to form on the surface of a magnetic fluid~\cite{Richter2005ferro-spikes}. The latter example directly informs our approach in this work. The experiments of Richter and Barashenkov~\cite{Richter2005ferro-spikes} utilise a probe coil to destabilise the flat state in a compact region to generate an axisymmetric soliton (see~\cite{Hill2020Localised} for a previous study of axisymmetric ferrofluid spikes in the absence of spatial heterogeneity). The induced spike persists following the removal of the probe coil, allowing the authors to study the behaviour of localised solutions to the spatially homogeneous problem. The Swift--Hohenberg equation we study in this work is taken to be a toy model for the probe coil problem, where we introduce a compactly-supported potential function to mimic the linear stability problem in the ferrofluid experiment. We note that there has been significant recent interest in the effects of spatial heterogeneity on pattern formation (see~\cite{Jaramillo2019striped-phase,Avery2024GL-fronts,Goh2024slow-ramp,Avery2019growing-stripes,Goh2018quenching-stripes,Chen2021stripe-growth,Goh2023review,Jaramillo2015inhom-stripes,Jaramillo2023inhom-targets,Jaramillo2018} for stripe/target deformation, \cite{Bastiaansen2020inhom-pulse,Doelman2018inhom-localized,Doelman2016defect,kamphuis2025pattern,Kao2014heterogeneous,Ponedel2016ForcedSnaking} for pattern formation in parabolic equations, and~\cite{Brooks_2019,Knight2013,Derks2012} in hyperbolic equations); the present work fits well within this expanding literature, while also developing a new direction for the rigorous study of localised multi-dimensional patterns.

In this paper, we develop a new theory for the existence of large amplitude localised patterns with dihedral symmetry, where the spatial heterogeneity in the linear operator is linearly destabilising in a compact disc centred at the origin and stabilising everywhere else; see Figure~\ref{fig:Heterogeneity}. The key advantage of this approach is that for this type of spatial heterogeneity the essential spectrum is bounded away from the imaginary axis and allows the use of Fredholm operator theory, overcoming the fundamental difficulty in the spatially homogeneous case. One is then left with the problem of finding bifurcating point spectra; we show how this can be done where the linear operator is isotropic, though we believe that this requirement can be relaxed. Provided one can then locate bifurcating point spectra, one can then employ Crandall--Rabinowitz theory~\cite{Crandall1971Bifurcation} and global bifurcation theory~\cite{Rabinowitz1971Global,BuffoniToland2003Global,Constantin2016} to prove the existence of branches of large-amplitude, nonlinear fully localised patterns. 

\begin{figure}[t!]
    \centering
    \includegraphics[width=0.7\linewidth]{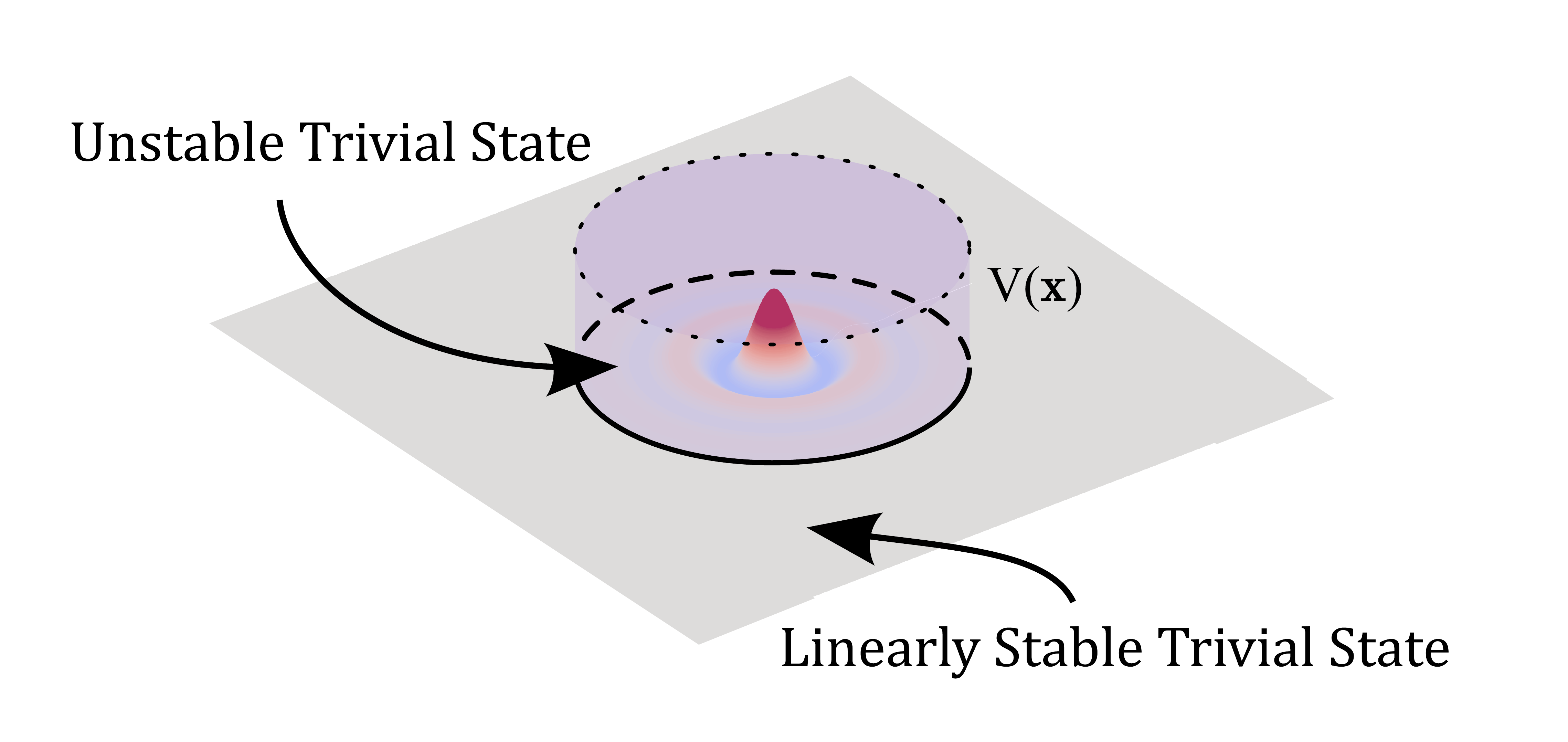}
    \caption{
    We consider a PDE system appended with a localised potential $V(\mathbf{x})$, causing the otherwise-stable trivial state to destabilise in a compact region and inducing the emergence of localised patterns.}
    \label{fig:Heterogeneity}
\end{figure}

The Crandall--Rabinowitz theorem allows one to prove the existence of small amplitude solutions bifurcating from a given base state while analytic global bifurcation theory---originally due to Dancer~\cite{Dancer1973Global} and later developed by Buffoni and Toland \cite{BuffoniToland2003Global} and Constantin, Strauss and V\u{a}rv\u{a}ruc\u{a}~\cite{Constantin2016}---allows one to continue small amplitude solutions to large amplitude. The Crandall--Rabinowitz theorem is typically applied to bounded domains or spatially periodic solutions; while localised solutions have been considered in some cases~\cite{Wheeler2015Solitary,Truong2022Global}, the application of Crandall--Rabinowitz theory is more difficult due to a lack of standard embedding and compactness properties in infinitely-extended spatial domains. While analytic global bifurcation theory has become invaluable in proving the existence of large amplitude water waves, including some recent applications to solitary waves (see~\cite[Section 5.6]{Haziot2022} and references therein), very few results exist in the pattern formation literature (see for instance~\cite{Nishiura1982Global,Fujii1983Global,Baltaev2008Global,garciaazpeitia2025globalbifurcationspiralwave} which use the global bifurcation theory of Rabinowitz~\cite{Rabinowitz1971Global}, and the more recent work of \cite{Bruell_2024,boehmer2025patternformationfilmrupture} which use analytic global bifurcation theory). A key difficulty with applying the theory of Buffoni and Toland~\cite{BuffoniToland2003Global} to localised solutions is a possible lack of compactness of solution sets. While in bounded domains (or, equivalently, for periodic solutions) local compactness of the solution set is guaranteed by a combination of Schauder estimates and Rellich--Kondrachev-type embeddings, these results no longer hold in infinitely extended domains. A crucial change in this approach was provided by Chen, Walsh and Wheeler in \cite{Chen2024Global}, where the authors adapted the analytic global continuation theorem of Buffoni and Toland so that local compactness was no longer an assumption, but rather a property of the solution set that could fail. To the best of the authors' knowledge, there are no global bifurcation results for localised patterns in pattern-forming systems, and so this paper presents a new application of the theory for 2D fully localised patterns with dihedral symmetry.

Our approach is to consider the following semi-linear PDE system. Fix $d\in\mathbb{N}$ with $d>1$, $N\in\mathbb{N}$, and consider 
\begin{equation}\label{e:geneq}
\partial_{t} u = \mathcal{L}_{\varepsilon}(\Delta)u + V_{\varepsilon}(\mathbf{x})\,u + \mathcal{N}_{\varepsilon}(u), \qquad\qquad u\in H^{2m}(\mathbb{R}^d;\mathbb{R}^N),  \qquad \varepsilon\in \Lambda\subseteq\mathbb{R},
\end{equation}
with $\mathcal{L}_{\varepsilon}(\Delta): H^{2m}(\mathbb{R}^d;\mathbb{R}^N) \to L^2(\mathbb{R}^d;\mathbb{R}^N)$ an isotropic linear operator (where $\mathcal{L}_{\varepsilon}(\cdot)$ is given by an $m$'th order polynomial, $\Delta$ denotes the $d$-dimensional Laplace operator, and $m>d/4$), $V_{\varepsilon}$ a localised potential, and $\mathcal{N}_\varepsilon$ a real-analytic function with $\mathcal{N}_{\varepsilon}(0) = \mathcal{N}_{\varepsilon}'(0) = 0$ for all $\varepsilon\in\Lambda$. We additionally assume that $\mathcal{L}_{\varepsilon}, V_{\varepsilon}, \mathcal{N}_{\varepsilon}$ are each analytic in $\varepsilon\in\Lambda$, and that the following hypotheses are satisfied.

\begin{hyp*} Let $X \subseteq H^{2m}(\mathbb{R}^{d};\mathbb{R}^{N})$\footnote{One must restrict to some subspace $X$ of $H^{2m}(\mathbb{R}^{d};\mathbb{R}^{N})$ in order to remove innate symmetries of the problem and obtain a one-dimensional kernel in hypothesis $(ii)$.} be an invariant subspace of \eqref{e:geneq} and (with slight abuse of notation) replace $\mathcal{L}_{\varepsilon}(\Delta)$ with its restriction onto $X$. We then suppose that

    \begin{enumerate}[label=(\roman*)]
    \item $\mathcal{L}_{\varepsilon}(-|\mathbf{k}|^2)<0$ for all $\mathbf{k}\in\mathbb{R}^d$, $\varepsilon\in \Lambda$
    ;
    \item there exists some $\varepsilon_0 \in \Lambda$ such that $\mathrm{ker}\,(\mathcal{L}_{\varepsilon_0}(\Delta) + V_{\varepsilon_0}) = \mathrm{span}\,\{v_0\}$ with $v_0\in X$; and
    \item 
    $\partial_\varepsilon(\mathcal{L}_{\varepsilon}(\Delta) + V_\varepsilon)_{\varepsilon = \varepsilon_0}v_0 \notin\mathrm{Im}(\mathcal{L}_{\varepsilon_0}(\Delta) + V_{\varepsilon_0})$. 
\end{enumerate}
\end{hyp*}

Hypothesis (i) supposes that the operator $\mathcal{L}_{\varepsilon}(\Delta)$ and the linearisation of \eqref{e:geneq} about some $u=\phi\in H^{2m}(\mathbb{R}^d;\mathbb{R}^N)$ are Fredholm operators with index zero. In particular, the essential spectrum of each operator is bounded away from zero and so any change in the linear stability of solutions to \eqref{e:geneq} is driven entirely by eigenvalues crossing the imaginary axis. Hypothesis (ii) supposes that, for some parameter value $\varepsilon = \varepsilon_0$, there is one eigenvalue lying on the imaginary axis with a single corresponding eigenfunction $v_0$. Finally, hypothesis (iii) supposes that the equation \eqref{e:geneq} has nondegenerate dependence on the parameter $\varepsilon$ at $\varepsilon = \varepsilon_0$. 

These hypotheses together tell us that the change in stability at $\varepsilon=\varepsilon_0$ corresponds to a transverse intersection between the trivial state $u\equiv0$ and some non-trivial bifurcating solution $u = u_*(\mathbf{x})$. We are then able to apply local and global bifurcation results to prove the existence of non-trivial fully localised patterns along a continuous bifurcation curve
\begin{equation*}
    \mathcal{C} = \{(\varepsilon,u) = (\mathcal{E}(s),v(s)):\; s\in\mathbb{R}\}.
\end{equation*}
Here, $\mathcal{E}(0) = \varepsilon_0$, $v(0)=0$, $v'(0) = v_0$, and the curve $\mathcal{C}$ possesses a local analytic parametrisation around each of its points. The local and global bifurcation results are proven in an open subspace $(\varepsilon,u)\in U\subseteq \Lambda\times X$, and so a solution $u=v(s)$ of \eqref{e:geneq} will remain spatially localised along the curve $\mathcal{C}$.

To illustrate our method, we consider the two-dimensional Swift--Hohenberg equation given by
\begin{equation}\label{e:SH}
\partial_t u = -[(1+\Delta)^2 + \varepsilon^2] u + 2\,\varepsilon^2\,V(|\mathbf{x}|)\,u + f(u), \qquad\quad \text{with}\qquad V(r) = \begin{cases}
    1, & r<R,\\
    0, & r>R.
\end{cases}
\end{equation}
where $u = u(t,\mathbf{x})\in\mathbb{R}$, $\mathbf{x}\in\mathbb{R}^2$, $R>0$ denotes the width of the potential $V$, and $f$ is real-analytic with $f(0) = f'(0) = 0$. This equation naturally destabilises the trivial state in a finite circle of radius $R$ while outside this region the trivial state is linearly stable. While \eqref{e:SH} is a considerably simpler model than the general PDE system \eqref{e:geneq}, the theory presented in later sections either follows identically or can be extended to \eqref{e:geneq} (under the assumption that hypotheses (i)-(iii) are satisfied); see Remark~\ref{rmk:generality} at the end of this section for further discussion of the generality of this approach.

We seek steady solutions $u = u(\mathbf{x})$ in $H^4_{\mathrm{e}}$---the subspace of functions in $H^4(\mathbb{R}^2)$ which are also even in $y$, i.e. $u(x,y) = u(x,-y)$---where we solve the steady equation
\begin{equation}\label{e:SH-steady}
0 = -[(1+\Delta)^2 + \varepsilon^2] u + 2\,\varepsilon^2\,V(|\mathbf{x}|)\,u + f(u),
\end{equation}
posed in the analogous even subspace $L^2_{\mathrm{e}}$ of $L^2(\mathbb{R}^2)$. We categorise possible solutions to \eqref{e:SH-steady} by introducing the following \textit{rotational subspaces}
\begin{equation}\label{def:Xmk}
\begin{split}
    X^{m}_{k} :={}& \left\{u\in H^{m}_{\mathrm{e}}\,:\; u(r\cos(\theta+\tfrac{2\pi}{k}),r\sin(\theta+\tfrac{2\pi}{k})) = u(r\cos\theta,r\sin\theta) \quad \forall \,r>0, \;\theta\in\mathbb{T}^1\right\},\\
\end{split}
\end{equation}
for any $k\in\mathbb{N}$ and $m\in\mathbb{N}_{0}$ (see Remark~\ref{rmk:Xmk} for the case when $k=0$). We note that we can equivalently define the spaces $X^m_k$
as the intersections of $H^m(\mathbb{R}^2;\mathbb{R})$ Sobolev spaces and $\mathbb{D}_{k}$ dihedral symmetry spaces, where we fix the dihedral line of reflection symmetry to be the x-axis.
\begin{rmk}\label{rmk:Xmk}
    We define, with a slight abuse of notation, the space $X^m_{0}$ to be the subspace of $H^m(\mathbb{R}^2;\mathbb{R})$ containing axisymmetric functions $u(\mathbf{x}) = u(|\mathbf{x}|)$; we likewise define $\mathbb{D}_0$ to be the set of axisymmetric functions in this notation. We also note that $X^m_{1}$ is the entire space $H^m_{\mathrm{e}}$ but we maintain the vacuous notation $X^m_1$ and $\mathbb{D}_{1}$ in order to be consistent with other values of $k\in\mathbb{N}_0$.
\end{rmk}

\begin{rmk}
    The rotational subspaces form chains of subspaces
\begin{equation*}
    X^{m+n}_{k} \subseteq X^{m}_k, \qquad\qquad\qquad X^{m}_{n\,k} \subseteq X^{m}_k
\end{equation*}
for any $k,m,n\in\mathbb{N}_0$. Furthermore, the dihedral space $\mathbb{D}_{k}$ represents an invariant subspace of the equation \eqref{e:SH} for each $k\in\mathbb{N}_0$.
\end{rmk}

For our particular choice of spatial heterogeneity, we are able to derive explicit characterisations for the point spectrum of the linearisation of \eqref{e:SH} about the trivial solution. In particular, we determine that eigenvalues $\lambda=\varepsilon^2\widetilde{\lambda}$ solve the implicit relation
\begin{equation*}
F_k(R,\varepsilon,\tilde{\lambda}) = \mathrm{Re}\left[\mathrm{e}^{\mathrm{i}\sin^{-1}(\widetilde{\lambda})}\, W[J_{k}(\alpha_{+}r), H^{(1)}_{k}( \beta r)](R)\,W[J_{k}(\alpha_{-}r), \overline{H^{(1)}_{k}( \beta r)}](R)\right]=0
\end{equation*}
for some $k\in\mathbb{N}_0$ and fixed $\varepsilon, R>0$. Here, $\alpha_{\pm} := \sqrt{1 \pm \varepsilon\,\sqrt{1 - \widetilde{\lambda}}}$, $\beta := \sqrt{1 + \mathrm{i}\,\varepsilon\,\sqrt{1 + \widetilde{\lambda}}}$; $J_k, H_k^{(1)}$ are the $k$'th order Bessel functions of the first and third kind; and $W[\cdot,\cdot]$ denotes the weighted Wronskian function given by
\begin{equation*}
    W[u,v](r) := r\left( u(r)\,v'(r) - u'(r)\,v(r)\right).
\end{equation*}
We numerically observe that there are infinitely many bifurcation points $\varepsilon = \varepsilon_{k,n}$ for $k\in\mathbb{N}_0$ and $n\in\mathbb{N}$, and that the parameter $R$ (corresponding to the width of the potential $V$) can be chosen such that the above hypotheses are satisfied.

\begin{figure}[htb!]
    \centering
    \includegraphics[width=\linewidth]{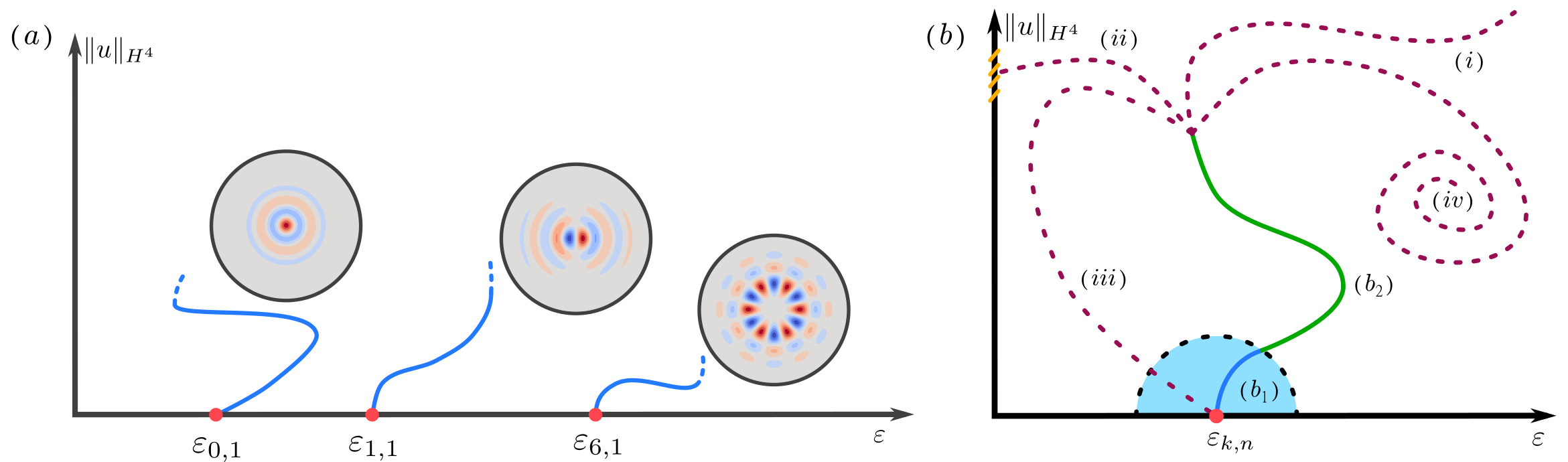}
    \caption{(a) A schematic bifurcation diagram showing different types of localised dihedral patterns bifurcating off the trivial state. (b) A local bifurcation curve ($b_1$) can be extended to a global bifurcation curve ($b_2$) using analytic bifurcation theory. In particular, the analytic curves that emanate from the local bifurcation either: $(i)$ blow up in their norm, $(ii)$ collide with the boundary $\varepsilon=0$, $(iii)$ form closed loops, or $(iv)$ lose compactness.}
    \label{fig:Bif_Schematic}
\end{figure}

Our aim is then to prove the schematic picture shown in Figure~\ref{fig:Bif_Schematic}(a). We note that we prove the existence of fully localised patterns with $\mathbb{D}_k$ symmetry for any $k\in\mathbb{N}_0$, and not just the $\mathbb{D}_0$, $\mathbb{D}_1$ and $\mathbb{D}_6$ localised patterns shown in Figure~\ref{fig:Bif_Schematic}(a). We employ analytic bifurcation theory in order to prove the existence of the bifurcating localised patterns, as illustrated in Figure~\ref{fig:Bif_Schematic}(b); we first prove the existence of a unique non-trivial solution branch in a neighbourhood of the bifurcation point (Figure~\ref{fig:Bif_Schematic}($b_1$)), which we call the \textit{local branch}, before extending solutions to large amplitude (Figure~\ref{fig:Bif_Schematic}($b_2$)), which we call the \textit{global branch}. The global solution branch is either unbounded, approaches the $\varepsilon=0$ axis, forms a closed loop, or loses compactness, as shown in Figure~\ref{fig:Bif_Schematic}($i$)-($iv$). In the case of radially symmetric patterns, we are able to prove that the global solution branch never loses compactness; an equivalent result for non-radially symmetric patterns remains an open problem.

We present two main theorems for the existence and stability of fully localised 2D pattern solutions to~\eqref{e:SH}. We begin by informally stating our existence result, where we refer to Theorems~\ref{thm:Local} and \ref{thm:Global} for the precise statements. 

\begin{rslt}[Existence]\label{rslt:1}
    Fix $k\in\mathbb{N}_0$ and $R>0$. For any $\varepsilon_0$ that is a simple root of $F_k(R,\varepsilon_0,0)=0$ (and is not a root for any other $k\in\mathbb{N}_0$), there exists a localised solution $u\in X^4_k$ which bifurcates from $(\varepsilon, u) = (\varepsilon_0, 0)$ and continues to a global solution branch.
\end{rslt}

Amongst the many bifurcating solutions we are considering, we are particularly interested in the initial destabilisation of the trivial state. We call this point the \textit{primary bifurcation}, corresponding to the the minimum value of $\varepsilon$ that solves $F_k(R,\varepsilon,0)=0$ for some $k\in\mathbb{N}_0$. We will presently observe that the primary branch consists of qualitatively different solutions for different parameter values, alternating between axisymmetric and non-axisymmetric profiles as $R$ increases. 

We prove the linear (in)stability of localised solutions to \eqref{e:SH} near their bifurcation, where we note qualitative differences between the primary and non-primary bifurcation points. We again present an informal version of our stability result here and refer the reader to Lemma~\ref{lem:Stab} for more details.

\begin{rslt}[Stability]\label{rslt:2}
Near the primary bifurcation point, the bifurcating solution is linearly stable if $\varepsilon>\varepsilon_0$ and unstable if $\varepsilon<\varepsilon_0$. Subsequent bifurcation points only yield locally unstable branches.
\end{rslt}

We briefly discuss how Results~\ref{rslt:1} and \ref{rslt:2}---proven for the Swift--Hohenberg equation \eqref{e:SH}---can be extended to more general PDE systems of the form \eqref{e:geneq}.
\begin{rmk}\label{rmk:generality}
Suppose we consider fully localised solutions to \eqref{e:geneq}, where we assume
\begin{itemize}
    \item $\mathcal{L}_{\varepsilon}(\cdot)$ is given by an $m$'th order polynomial, so that $\mathcal{L}_{\varepsilon}(\Delta): H^{2m}(\mathbb{R}^d;\mathbb{R}^N) \to L^2(\mathbb{R}^d;\mathbb{R}^N)$;
    \item $\mathcal{N}_\varepsilon$ is a real-analytic function with $\mathcal{N}_{\varepsilon}(0) = \mathcal{N}_{\varepsilon}'(0) = 0$ for all $\varepsilon\in\Lambda$, so that the linearisation of \eqref{e:geneq} is given by $(\mathcal{L}_\varepsilon(\Delta) + V_\varepsilon)\,u$;
    \item $m>d/4$, so that $H^{2m}(\mathbb{R}^d;\mathbb{R}^N)$ is a Banach algebra and hence $\mathcal{N}_\varepsilon:H^{2m}(\mathbb{R}^d;\mathbb{R}^N) \to L^{2}(\mathbb{R}^d;\mathbb{R}^N)$ analytically;
    \item $V_\varepsilon$ is chosen so that $\lim_{|\mathbf{x}|\to\infty}V_\varepsilon(\mathbf{x}) = 0$, and the multiplication operator $u\mapsto V_\varepsilon\,u$ maps from $H^{2m}(\mathbb{R}^d;\mathbb{R}^N)$ to $L^{2}(\mathbb{R}^d;\mathbb{R}^N)$; and
    \item $\mathcal{L}_{\varepsilon}, V_{\varepsilon}, \mathcal{N}_{\varepsilon}$ are each analytic in $\varepsilon\in\Lambda$.
\end{itemize}
Then, hypotheses (i-iii) guarantee that \eqref{e:geneq} satisfies the conditions for a Crandall-Rabinowitz local bifurcation. Since any solution $\phi\in H^{2m}(\mathbb{R}^d;\mathbb{R}^N)$ decays to zero as $|\mathbf{x}|\to\infty$, the linearisation about $\phi$ possesses the same Fredholm properties as $\mathcal{L}_{\varepsilon}(\Delta)$ which, coupled with the analytic dependence of \eqref{e:geneq} on $(\varepsilon,u)$, allows one to also globally continue the local bifurcation curves to large amplitude. For radially-symmetric patterns, our proof that global branches do not lose compactness also applies directly to \eqref{e:geneq} under the assumptions above.

The key difficulty in considering the general PDE system \eqref{e:geneq} is verifying hypotheses (ii) and (iii). Our choice of a piecewise-constant, radially-symmetric potential $V$ in \eqref{e:SH} allows for the explicit analysis of the point spectrum of $\mathcal{L}$ in Lemma~\ref{lem:eig-2}; even extending our analysis to a potential without radial symmetry results in a significantly more difficult eigenvalue problem, which may not be analytically tractable. 
\end{rmk}

    The paper is outlined as follows. In \S\ref{s:spec} we characterise the linear spectrum of \eqref{e:SH} and prove certain properties required for our bifurcation results. We then present our local and global bifurcation theorems in \S\ref{s:bif}, where we also prove the type of bifurcation and local stability of the bifurcating solutions. We finally conclude and discuss open problems in \S\ref{s:discuss}.

\section{Spectral properties of the linear operator}\label{s:spec}
We begin by considering the linear stability of solutions to \eqref{e:SH}. In particular, we analyse the stability of a general solution $\phi\in H^{4}_{\mathrm{e}}$---including the trivial solution $\phi\equiv0$---with respect to continuous perturbations, thus characterising the \textit{essential spectrum} of the linear problem (\S\ref{s:spec-ess}). We then characterise the \textit{point spectrum} of the linear problem by constructing the eigenvalues and critical eigenfunctions that determine the stability of the trivial solution (\S\ref{s:spec-point}).

We consider the linear operator $\mathcal{L}$, given by
\begin{equation}\label{def:L}
        \mathcal{L}u = -(1+\Delta)^2 u - \varepsilon^2\,u + 2\,\varepsilon^2\,V(|\mathbf{x}|)\,u
\end{equation}
as well as the linearisation $\mathcal{L}_{\phi}$ of \eqref{e:SH} about a function $\phi\in H^4_{\mathrm{e}}$, given by
\begin{equation}\label{def:L-phi}
        \mathcal{L}_{\phi}u = -(1+\Delta)^2 u - \varepsilon^2\,u + 2\,\varepsilon^2\,V(|\mathbf{x}|)\,u + f'(\phi(\mathbf{x}))\,u.
\end{equation}

\subsection{Essential spectrum}\label{s:spec-ess}
We note that \eqref{def:L} is a special case of \eqref{def:L-phi} with $\phi\equiv0$, and so we present the following result for the general linear operator $\mathcal{L}_{\phi}$ with $\phi\in H^4_{\mathrm{e}}$.

\begin{lem}\label{lem:Fredholm;L}
    Fix $\varepsilon,R>0$ and suppose that $\phi\in H^4_{\mathrm{e}}$. The linear operator $\mathcal{L}_{\phi}:H^{4}_{\mathrm{e}}\to L^2_{\mathrm{e}}$ given by \eqref{def:L-phi} is a Fredholm operator with index $0$. In particular, $\sigma_{\mathrm{ess}}(\mathcal{L}_\phi) =(-\infty,-\varepsilon^2]$, where $\sigma_{\mathrm{ess}}(\mathcal{L}_\phi)$ denotes the essential spectrum of $\mathcal{L}_\phi$.
\end{lem}
\begin{proof}
We first note that the operator $\mathcal{L}_\phi$ can be decomposed into three operators $\mathcal{T},\mathcal{K}_V,\mathcal{K}_{\phi}: H^{4}_{\mathrm{e}}\to L^2_{\mathrm{e}} $ by $\mathcal{L}_\phi = \mathcal{T} + \mathcal{K}_V + \mathcal{K}_\phi$, where
\begin{equation*}
    \mathcal{T}u = -[(1+\Delta)^2 + \varepsilon^2]\,u,\qquad\qquad \mathcal{K}_V\,u = 2\,\varepsilon^2\,V(|\mathbf{x}|)\,u, \qquad\qquad \mathcal{K}_{\phi}u = f'(\phi(\mathbf{x}))\,u.
\end{equation*}
Using the fact that $V(|\mathbf{x}|)\,u \in H^{4}(B_{C})$, where $B_{C}$ denotes the open ball of radius $C$ about the origin, and applying the Rellich--Kondrachov theorem, we conclude that $\mathcal{K}_V:H^{4}_{\mathrm{e}} \to L^2_{\mathrm{e}}$ is a compact operator. Likewise, since $f$ is real analytic and $\phi\in H^{4}_{\mathrm{e}}$, we conclude that $f'(\phi)\in H^{4}_{\mathrm{e}}$ and thus, by Lemma~\ref{lem:compact}, $K_{\phi}:H^{4}_{\mathrm{e}}\to L^2_{\mathrm{e}}$ is also compact. It follows that $\mathcal{L}_\phi$ is Fredholm if and only if $\mathcal{T}$ is Fredholm. Using the two-dimensional Fourier transform, we note that the operator $\mathcal{T}$ is unitarily equivalent to multiplication by its symbol
\begin{equation*}
    \widehat{\mathcal{T}} = -(1-|\mathbf{k}|^2)^2 - \varepsilon^2, \qquad\qquad \mathbf{k}\in\mathbb{R}^2
\end{equation*}
and so we conclude that the operators $\mathcal{T}$ and $\mathcal{L}_\phi$ are Fredholm with index $0$ for all $\varepsilon>0$. Since the essential spectrum is preserved under compact perturbations, it follows that $\sigma_{\mathrm{ess}}(\mathcal{L}_\phi) = \sigma_{\mathrm{ess}}(\widehat{\mathcal{T}})=(-\infty,-\varepsilon^2]$.
\end{proof}

\begin{rmk}
    The linear operator $\mathcal{L}_2 : H^{2}_{\mathrm{e}}\times H^{2}_{\mathrm{e}}\to L^2_{\mathrm{e}} \times L^2_{\mathrm{e}}$, given by
\begin{equation*}
    \begin{split}
        \mathcal{L}_{2}\mathbf{u} ={}& \begin{pmatrix} \Delta +1 & -1 \\   \varepsilon^2 - 2\,\varepsilon^2\,V(|\mathbf{x}|) & \Delta + 1\end{pmatrix}\mathbf{u},
    \end{split}
\end{equation*}
    is isomorphic to $\mathcal{L}$ by the mapping $\mathcal{S}: H^{4}_{\mathrm{e}}\to H^{2}_{\mathrm{e}} \times H^{2}_{\mathrm{e}},\, u\mapsto (u, (1+\Delta)u)$, and thus is also a Fredholm operator with index zero. One would thus expect the results presented herein to extend to reaction-diffusion systems without significant changes.
\end{rmk}

\subsection{Point spectrum}\label{s:spec-point}
In the previous section, we determined that the linear operator $\mathcal{L} : H^{4}_{\mathrm{e}} \to L^2_{\mathrm{e}}$ is a Fredholm operator with index $0$ for $\varepsilon>0$ and, in particular, that $\sigma_{\mathrm{ess}}(\mathcal{L})=(-\infty,-\varepsilon^2]$. This tells us that the essential spectrum of $\mathcal{L}$ is bounded away from $0$, and so any changes in the linear stability of $\mathcal{L}$ must be driven by eigenvalues crossing the imaginary axis. 

We now characterise the point spectrum of $\mathcal{L}$ by considering the linear eigenvalue problem $\lambda u = \mathcal{L}\,u$, where $u\in H^{4}_{\mathrm{e}}$ is an eigenfunction of $\mathcal{L}$ with associated eigenvalue $\lambda\in\mathbb{C}$. For notational simplicity, we introduce some $\tilde{\lambda}\in\mathbb{C}$, with $\lambda = \tilde{\lambda}\,\varepsilon^2$, and consider the scaled eigenvalue problem
\begin{equation}\label{e:Stab}
    \varepsilon^2\,\tilde{\lambda}\, u = -(1+\Delta)^2 u - \varepsilon^2 u + 2\,\varepsilon^2\,V(|\mathbf{x}|)\,u.
\end{equation}

\begin{lem}\label{lem:eig}
    Fix $R>0$ and $\varepsilon>0$. The eigenvalue problem \eqref{e:Stab} possesses a separable solution $u = u_k \in H^{4}_{\mathrm{e}}$ given by $u_k(r\cos\theta,r\sin\theta) = v_k(r)\,\cos(k\theta)$ with
\begin{equation*}
v_k(r) = \begin{cases}
\displaystyle \cos(\phi + \psi)\frac{J_{k}(\alpha_{+}\,r)}{J_{k}(\alpha_{+}\,R)} + \cos(\phi - \psi)\frac{J_{k}(\alpha_{-}\,r)}{J_{k}(\alpha_{-}\,R)}, & r<R,\\
\displaystyle \sqrt{2}\,\sqrt{1 - \tilde{\lambda}}\, \mathrm{Re}\left(\mathrm{e}^{\mathrm{i}\phi}\, \frac{H^{(1)}_{k}( \beta r)}{H^{(1)}_{k}(\beta R)}\right), & r>R,\\
\end{cases}  
\end{equation*}
where $\alpha_{\pm} := \sqrt{1 \pm \varepsilon\,\sqrt{1 - \tilde{\lambda}}}$, $\beta := \sqrt{1 + \mathrm{i}\,\varepsilon\,\sqrt{1 + \tilde{\lambda}}}$; $J_k, H_k^{(1)}$ are the $k$'th order Bessel functions of the first and third kind; $\psi$ is given by $\psi = \frac{1}{2}\sin^{-1}(\tilde{\lambda}) + \frac{\pi}{4}$; $\phi$ is given by
    \begin{equation*}
\mathrm{e}^{2\mathrm{i}\phi } = \mathrm{i}\,\mathrm{e}^{-\mathrm{i}\sin^{-1}(\tilde{\lambda})}\frac{W[J_{k}(\alpha_{+}r), \overline{H^{(1)}_{k}( \beta r)}](R)\,H^{(1)}_{k}( \beta R)}{W[J_{k}(\alpha_{+}r), H^{(1)}_{k}( \beta r)](R)\,\overline{H^{(1)}_{k}( \beta R)}};
    \end{equation*}
    $W[\cdot,\cdot]$ denotes the weighted Wronskian function given by
\begin{equation*}
    W[u,v](r) := r\left( u(r)\,v'(r) - u'(r)\,v(r)\right);
\end{equation*}
    and $\tilde{\lambda} = \tilde{\lambda}(R, \varepsilon)$ satisfies
\begin{equation*}
F_k(R,\varepsilon,\tilde{\lambda}) = \mathrm{Re}\left[\mathrm{e}^{\mathrm{i}\sin^{-1}(\tilde{\lambda})}\, W[J_{k}(\alpha_{+}r), H^{(1)}_{k}( \beta r)](R)\,W[J_{k}(\alpha_{-}r), \overline{H^{(1)}_{k}( \beta r)}](R)\right]=0
\end{equation*}
with $\tilde{\lambda}\notin(-\infty,-1]$.
\end{lem}

\begin{figure}[t!]
    \centering
    \includegraphics[width=\linewidth]{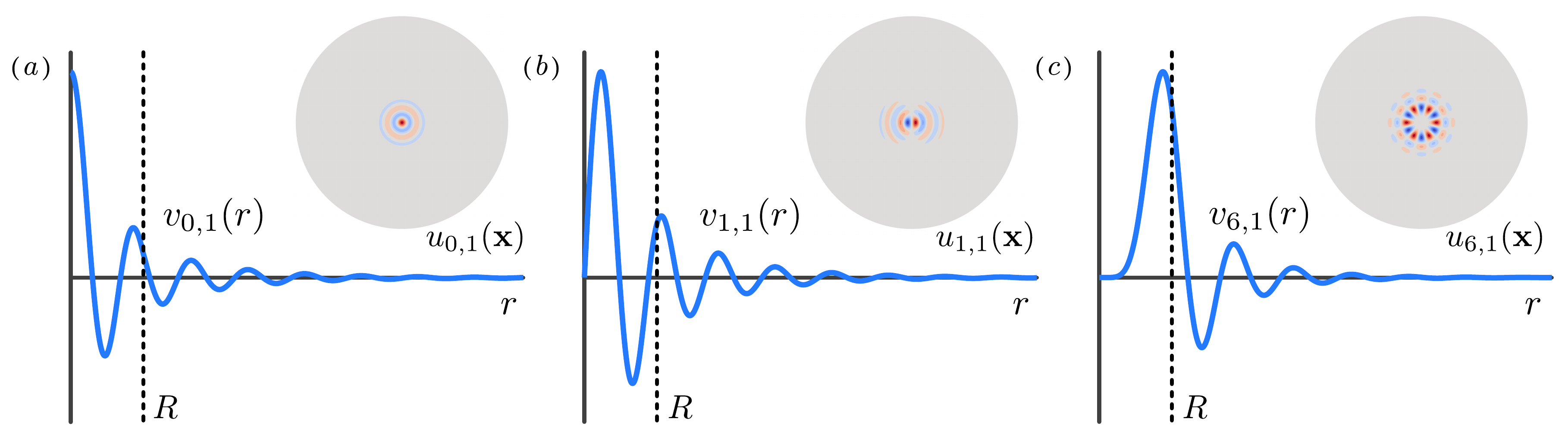}
    \caption{Radial profiles (resp. planar profiles) of linear eigenfunctions $v_{k,n}(r)$ ($u_{k,n}(\mathbf{x})$) are plotted for $R=8$ and (a) $\varepsilon = \varepsilon_{0,1}$, (b) $\varepsilon = \varepsilon_{1,1}$, and (c) $\varepsilon = \varepsilon_{6,1}$.}
    \label{fig:linear_eigs}
\end{figure}

The proof of this lemma is given in Appendix~\ref{app:spec}; we illustrate examples of the linear eigenfunctions $u = u_{k,n}$ in Figure~\ref{fig:linear_eigs}. We highlight the respective $\mathbb{D}_0$ and $\mathbb{D}_1$ eigenfunctions $u_{0,1}$ and $u_{1,1}$, plotted in Figures~\ref{fig:linear_eigs} (a) and (b), which we refer to as the `axisymmetric spot' and `dipole' patterns. The equation $F_k(R, \varepsilon, \tilde{\lambda})=0$ implicitly defines eigenvalues $\tilde{\lambda} = \tilde{\lambda}(R,\varepsilon)$ for fixed $R,\varepsilon>0$, and the associated bifurcation equation $F_k(R, \varepsilon, 0)=0$ equivalently defines bifurcation points $\varepsilon = \varepsilon(R)$ for a fixed $R>0$. While we do not obtain an explicit form for the eigenvalues $\tilde{\lambda}$, we are able prove the following qualitative properties.

\begin{lem}\label{lem:eig-2}
    An eigenvalue $\tilde{\lambda}\in\mathbb{C}$ of \eqref{e:Stab} satisfies the following properties:
    \begin{enumerate}[label=(\roman*)]
        \item $\tilde{\lambda}\in\mathbb{R}$ and $-1 < \tilde{\lambda} \leq 1$ for all $R,\varepsilon>0$;
        \item $\tilde{\lambda}\to -1$ as $R\to0$ for any fixed $\varepsilon>0$; and
        \item for each $k\in\mathbb{N}_0$ and sufficiently large $R$ with $\varepsilon = \tilde{\varepsilon}\,R^{-1}$, there are countably many eigenvalues $\{\tilde{\lambda}_{k,n}\}_{n=1}^{\infty}$ which solve the implicit equation
        \begin{equation*}
            \tilde{\varepsilon} = \frac{2\sin^{-1}(\tilde{\lambda}_{k,n}) + (2n - 1)\pi}{2\,\sqrt{1 - \tilde{\lambda}_{k,n}}}.
        \end{equation*}
        The eigenvalue $\tilde{\lambda}_{k,n}$ is a monotone increasing function of $\tilde{\varepsilon}$ for each $n\in\mathbb{N}$.
    \end{enumerate}
\end{lem}
\begin{proof}
We begin by deriving simple \textit{a priori} estimates for an eigenvalue of \eqref{e:Stab}. The operator $\mathcal{L}$ is symmetric with respect to the complex $L^2$-inner product and so it follows that $\tilde{\lambda}\in\mathbb{R}$. Assuming that $u$ satisfies \eqref{e:Stab} for some $\tilde{\lambda}\in\mathbb{R}$, we take the $L^2$-inner product of \eqref{e:Stab} with $u$ and, applying Green's identities, obtain
\begin{equation}\label{id:apriori}
    \tilde{\lambda}  = 2\,\frac{\| V\,u\|^2}{\|u\|^2} -\frac{\|(1 + \Delta) u\|^2}{\varepsilon^2\,\|u\|^2}  - 1.
\end{equation}
It follows from the simple estimate $\| V\,u\|^2\leq\| u\|^2$ that $\tilde{\lambda} \leq 1$, where we recall that $\tilde{\lambda}>-1$ follows from Lemma~\ref{lem:eig}. Furthermore, we note that
\begin{equation*}
    \| V\,u\|^2 = \int_{B_{R}}|u(\mathbf{x})|^2\,\mathrm{d}\mathbf{x}  
    \to 0 \qquad \text{as $R\to0$}
\end{equation*}
and so there exists some $R_0(\varepsilon)>0$ such that $\tilde{\lambda} < 0$ for all $R\leq R_0$ and fixed $\varepsilon>0$.

\begin{figure}[t!]
    \centering
    \includegraphics[width=\linewidth]{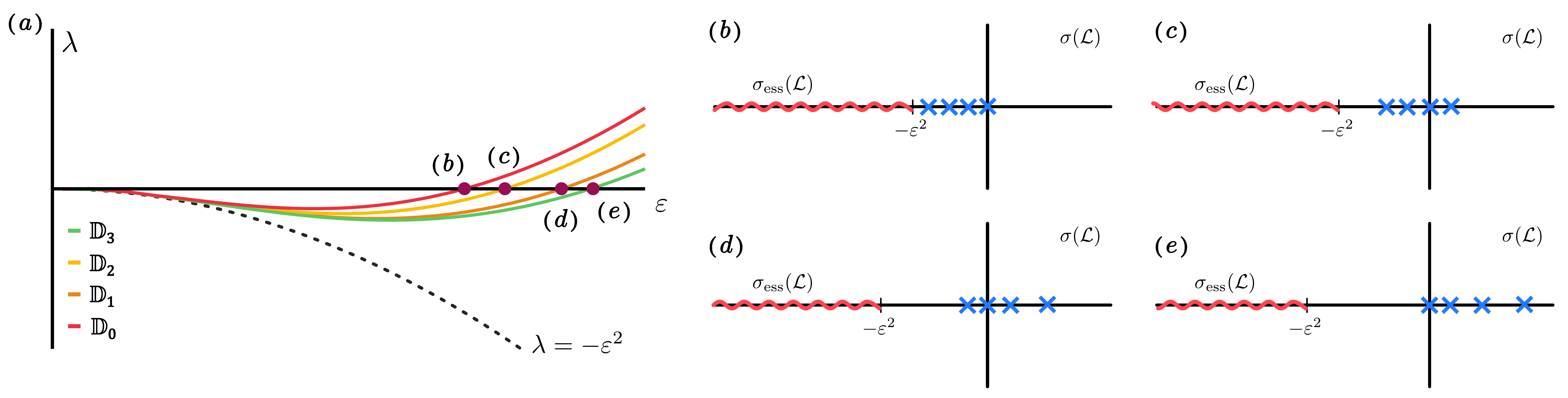}
    \caption{(a) The most unstable eigenvalues $\lambda$ of $\mathcal{L}$ for $k=0,1,2,3$ are numerically computed and plotted for $R=5$ and $\varepsilon>0$; (b) - (e) a cartoon of the spectrum of the linear operator $\mathcal{L}$ is presented for critical values of $\varepsilon$, where only the essential spectrum and the four most unstable eigenvalues are shown.}
    \label{fig:Stab_Plot}
\end{figure}

We recall from Lemma~\ref{lem:eig} that an eigenvalue $\tilde{\lambda}$ of \eqref{e:Stab} must satisfy
\begin{equation*}
    F_k(R,\varepsilon,\tilde{\lambda}) = \mathrm{Re}\left[\mathrm{e}^{\mathrm{i}\sin^{-1}(\tilde{\lambda})}\, W[J_{k}(\alpha_{+}r), H^{(1)}_{k}( \beta r)](R)\,W[J_{k}(\alpha_{-}r), \overline{H^{(1)}_{k}( \beta r)}](R)\right]=0
\end{equation*}
for some $k\in\mathbb{N}_0$, and we note the following asymptotic expansions:
\begin{enumerate}[label=(\alph*)]
    \item for fixed $\varepsilon>0$
\begin{equation*}
    F_k(R,\varepsilon,\tilde{\lambda}) = \frac{4}{\pi^2}\left(\frac{\alpha_{+}\,\alpha_{-}}{|\beta|^2}\right)^{k}\,\sqrt{1 - \tilde{\lambda}^2} + \mathcal{O}(R)
\end{equation*}
    as $R\to0$; and
    \item for fixed $\varepsilon=\tilde{\varepsilon}\,R^{-1}$ with fixed $\tilde{\varepsilon}>0$
\begin{equation*}
    F_k(R,\tilde{\varepsilon}\,R^{-1},\tilde{\lambda}) = \frac{4}{\pi^2}\,\mathrm{e}^{-\sqrt{1 + \tilde{\lambda}}\,\tilde{\varepsilon}}\,\cos\left(\sin^{-1}(\tilde{\lambda}) - \sqrt{1 - \tilde{\lambda}}\,\tilde{\varepsilon}\right) + \mathcal{O}(R^{-1})
\end{equation*}
    as $R\to\infty$.
\end{enumerate}
It follows that $\tilde{\lambda}^2 \to 1$ as $R\to0$ for fixed $\varepsilon>0$, where the \textit{a priori} estimates tell us that in fact $\tilde{\lambda}\to-1$.

In the case when $\varepsilon = \tilde{\varepsilon}\,R^{-1}$ and $R\to\infty$, the eigenvalue $\tilde{\lambda}$ satisfies
\begin{equation*}
    \tilde{\varepsilon} = g_{k,n}(\tilde{\lambda}) := \frac{2\,\sin^{-1}(\tilde{\lambda}) + (2n-1)\pi}{2\,\sqrt{1 - \tilde{\lambda}}}
\end{equation*}
for any $n\in\mathbb{N}$. We note that $\sin^{-1}(\tilde{\lambda})>-\frac{\pi}{2}$ for $\tilde{\lambda}\in(-1,1)$, and so $g_{k,n}(\tilde{\lambda})>0$ for each $n\in\mathbb{N}$. Furthermore, 
\begin{equation*}
    g_{k,n}'(\tilde{\lambda}) = \frac{1}{\sqrt{1 - \tilde{\lambda}^2}\,\sqrt{1 - \tilde{\lambda}}} + \frac{1}{2(1 - \tilde{\lambda})}\,g_{k,n}(\tilde{\lambda}) > \frac{1}{\sqrt{1 - \tilde{\lambda}^2}\,\sqrt{1 - \tilde{\lambda}}} > 0
\end{equation*}
and so $g_{k,n}$ is a strictly monotone increasing positive function. It follows that as $R\to\infty$ there is a countably infinite family of eigenvalues $\tilde{\lambda}_{k,n} = f_{k,n}(\tilde{\varepsilon}) := g_{k,n}^{-1}(\tilde{\varepsilon})$ for $n\in\mathbb{N}$. Additionally, the function $f_{k,n}$ is also strictly monotone increasing, and so each $\tilde{\lambda}_{k,n}$ increases as $\tilde{\varepsilon}$ increases.
\end{proof}

Property $(iii)$ in Lemma~\ref{lem:eig-2} states that, for $R$ sufficiently large, the eigenvalues of \eqref{e:SH-stab} are monotone increasing in $\varepsilon$; we also numerically observe this behaviour for moderate values of $R$, as shown in Figure~\ref{fig:Stab_Plot}(a). We note that, while $\tilde{\lambda}_{k,n}$ appears independent of the value of $k\in\mathbb{N}_0$ in the asymptotic regime, the eigenvalues will vary in both $n$ and $k$ for moderate values of $R$. The qualitative behaviour of the spectrum of $\mathcal{L}$ is illustrated in Figure~\ref{fig:Stab_Plot}(b-e), where the essential spectrum (given by $\sigma_{\mathrm{ess}}(\mathcal{L}) = (-\infty, -\varepsilon^2]$) moves to the left as $\varepsilon$ increases and point spectra move to the right, causing additional bifurcations as each eigenvalue crosses the imaginary axis.

We introduce the label $\varepsilon_{k,n}$ to be the value of $\varepsilon$ where $\tilde{\lambda}_{k,n}=0$. For sufficiently large values of $R$, we obtain that $\varepsilon_{k,n} = \frac{(2n-1)\pi}{2\,R}$ for all $k\in\mathbb{N}_0$ and $n\in\mathbb{N}$; we again numerically observe that each $\varepsilon_{k,n}$ persists for moderate values of $R$ and varies in both $n$ and $k$, as shown in Figure~\ref{fig:Roots}. In particular, we plot in Figure~\ref{fig:Roots}(a) the implicit curves for $F_k(R,\varepsilon,0)=0$ and observe that, while each implicit curve converges to a monotonic curve $\varepsilon_{k,n} = \frac{(2n-1)\pi}{2\,R}$ as $R\to\infty$, the curves exhibit greater oscillations for moderate values of $R$. 

In order to use Crandall--Rabinowitz bifurcation theory, we require that the kernel of $\mathcal{L}$ is one dimensional at the bifurcation point. This is equivalent to saying that, for a given value of $R$, the curves in Figure~\ref{fig:Roots}(a-b) do not intersect one another for all $k\in\mathbb{N}_0$. While it may be possible to prove this holds for certain parameter regimes, we instead present the following hypothesis which we motivate by Figure~\ref{fig:Roots}.

\begin{hyp}\label{hyp:ker}
We suppose that, for our choice of $R>0$ and $k\in\mathbb{N}_{0}$, the kernel of $\mathcal{L}:H^4_{\mathrm{e}}\to L^2_{\mathrm{e}}$ at $\varepsilon=\varepsilon_{k,n}$ is entirely spanned by the eigenfunction $u_{k,n}\in X^4_k$ given in Lemma~\ref{lem:eig}.
\end{hyp}

We present the following analytic predictions based on Lemma~\ref{lem:eig-2} and Figure~\ref{fig:Roots}, which we will verify numerically in Section~\ref{s:bif-loc}. We first note that as $R$ passes from $R=2.6$ to $R=2.8$, the primary bifurcating eigenfunction alternates from the axisymmetric spot $u_{0,1}$ to the dipole pattern $u_{1,1}$, plotted in Figures~\ref{fig:linear_eigs}(a) and (b), respectively. Furthermore, we note that Figure~\ref{fig:Roots}(b) suggests that one can choose a value of $R$ such that at a given bifurcation point $\varepsilon=\varepsilon_{k,n}$ the linear operator $\mathcal{L}$ has a one-dimensional kernel; in particular, we predict that $u_{6,1}$ (plotted in Figure~\ref{fig:linear_eigs}(c)) appears to satisfy this condition for $R=7.5$, as shown in Figure~\ref{fig:Roots}(d).

\begin{figure}[t!]
    \centering
    \includegraphics[width=\linewidth]{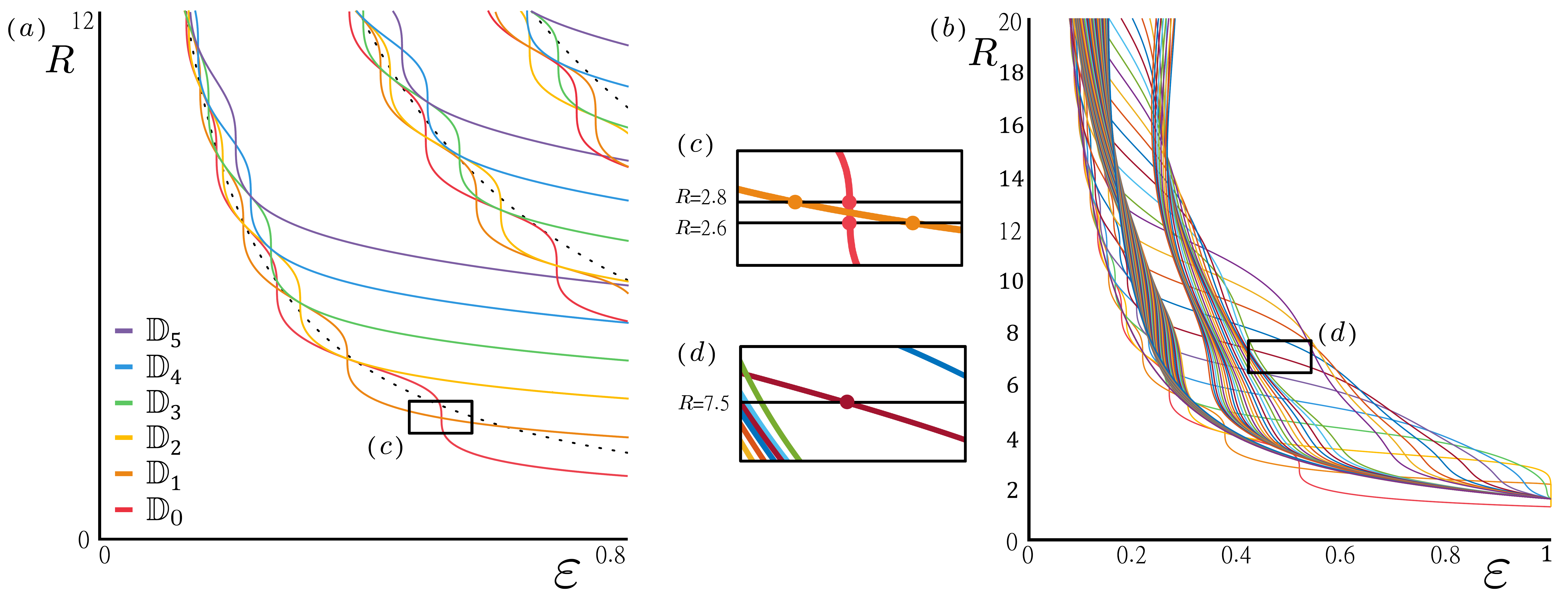}
    \caption{Implicit curves of $F_k(R,\varepsilon,0)=0$ for (a) $k=0,1,\dots,5$ and (b) $k=0,1,\dots,150$. The label $\mathbb{D}_k$ indicates the smallest dihedral group in which the linear eigenfunction $u = u_{k,n}$ lies, and the dotted black lines indicate the asymptotic curve $\varepsilon = \frac{(2n-1)\pi}{2\,R}$. (c) As $R$ passes between $2.6$ and $2.8$, the primary bifurcation swaps from the $\mathbb{D}_{0}$ pattern to the $\mathbb{D}_1$ pattern. (d) At $R=7.5$ the $\mathbb{D}_6$ pattern is an isolated root of $F_k(R,\varepsilon,0)=0$.}
    \label{fig:Roots}
\end{figure}

\section{Bifurcation results}\label{s:bif}

We now prove the existence of fully localised patterns bifurcating from trivial state. We formulate \eqref{e:SH} as the bifurcation problem
\begin{equation}\label{def:Fu}
0 = \mathscr{F}(\varepsilon, u) := -(1+\Delta)^2 u - \varepsilon^2 u + 2\,\varepsilon^2\, V(|\mathbf{x}|)\,u + f(u) 
\end{equation}
and prove the existence of a non-trivial solution curve $(\varepsilon, u) = (\mathcal{E}(s), v(s))$ parametrised by some $s\in\mathbb{R}$, first in a local neighbourhood of the bifurcation point (with $|s|\ll1$) in \S\ref{s:bif-loc}, before extending to a global solution curve of large amplitude patterns (with $s\in\mathbb{R}$) in \S\ref{s:bif-glob}.

\subsection{Local bifurcation}\label{s:bif-loc}
We begin by proving the existence of a unique bifurcation branch in a neighbourhood of the point $(\varepsilon, u)=(\varepsilon_{k,n},0)$, which we call the \textit{local bifurcation branch}. We will use an analytic version of the Crandall--Rabinowitz theorem (cf.\ Buffoni and Toland \cite[Theorem 8.3.1]{BuffoniToland2003Global}), which we state as Theorem~\ref{thm:CR} in the appendix. In order to prove that a local bifurcation occurs, we need to verify the following conditions:
\begin{enumerate}[label=$(\roman*)$]
    \item $\mathcal{L}: H^{4}_{\mathrm{e}} \to L^{2}_{\mathrm{e}}$ is a Fredholm operator of index zero at $\varepsilon=\varepsilon_{k,n}$,
    \item $\mathrm{ker}(\mathcal{L})= \mathrm{span}\{ u_{k,n} \}$ at $\varepsilon=\varepsilon_{k,n}$, and
    \item the transversality condition $\mathcal{P}(\mathrm{d}_1 \mathrm{d}_2 \mathscr{F}[\varepsilon_{k,n},0](1,u_{k,n})) \neq 0$ holds, where $\mathcal{P} : L^{2}_{\mathrm{e}} \to L^{2}_{\mathrm{e}}$ is a projection with $\mathrm{Im}(\mathcal{L})= \mathrm{ker}(\mathcal{P}).$
\end{enumerate}

Conditions $(i)$ and $(ii)$ were considered in Section~\ref{s:spec} and so it remains to prove that condition $(iii)$ holds. We begin by noting the following characterisation of the projection operator $\mathcal{P}$.
\begin{rmk}\label{rmk:Proj}
Suppose $\mathcal{L}$ is a symmetric operator with respect to some inner product $\langle\cdot,\cdot\rangle$; i.e,
\begin{equation*}
        \langle \mathcal{L} u_1, u_2 \rangle = \langle u_1, \mathcal{L} u_2 \rangle,\qquad\qquad \forall u_1,u_2\in X,
    \end{equation*}
    and $\mathrm{ker}(\mathcal{L}) = \mathrm{span}(v_0)$. There exists a projection $\mathcal{P}$ with $\mathrm{Im}(\mathcal{L}) = \mathrm{ker}(\mathcal{P})$ given by 
\begin{equation*}
        \mathcal{P}\,u = \frac{\langle v_0, u\rangle}{\langle v_0, v_0\rangle}\,v_0.
    \end{equation*}
\end{rmk}

We are now able to prove the following lemma regarding the transversality condition $(iii)$. 
\begin{lem}\label{lem:transversal}
Suppose that $\mathrm{ker}\,(\mathcal{L}) = \mathrm{span}\{v_0\}$ for some $\varepsilon_0>0$ and $v_0\in H^{4}_{\mathrm{e}}$. It follows that
\begin{equation*}
    \mathcal{P}(\mathrm{d}_1 \mathrm{d}_2 \mathscr{F}[\varepsilon_0,0](1,v_0)) = 0 \qquad\iff\qquad v_0 = 0.
\end{equation*}
\end{lem}
\begin{proof}
We prove this statement in two steps. We first show that $\mathcal{P}(\mathrm{d}_1 \mathrm{d}_2 \mathscr{F}[\varepsilon_0,0](1,v_0)) = 0$ is equivalent to $(1+\Delta)v_0 = 0$ for $v_0\in\mathrm{ker}\,(\mathcal{L})$, before then proving that $(1+\Delta)v_0 = 0$ if and only if $v_0=0$ for any $v_0\in H^4_{\mathrm{e}}$.

We note that $\mathcal{L}$ is symmetric with respect to the standard real-valued $L^2$ inner product; i.e.,
    \begin{equation*}
        \langle \mathcal{L} u_1, u_2 \rangle_{L^2} = \langle u_1, \mathcal{L} u_2 \rangle_{L^2},\qquad\qquad \forall u_1,u_2\in H^{4}_{\mathrm{e}}
    \end{equation*}
    and so it follows that we can define $\mathcal{P}$ as
    \begin{equation*}
        \mathcal{P}\,u = \frac{\langle v_0, u \rangle_{L^2}}{\| v_0 \|^2_{L^2}}\,v_0,
    \end{equation*}
    as stated in Remark~\ref{rmk:Proj}. Noting that $\mathrm{d}_1 \mathrm{d}_2 \mathscr{F}[\varepsilon_0,0](1,v) = -2\varepsilon_0\,[-1 + 2\,V(|\mathbf{x}|)]\,v = \frac{2}{\varepsilon_0}(\mathcal{L} + (1+\Delta)^2)v$, we obtain
    \begin{equation*}
    \begin{split}
        \mathcal{P}\left(\mathrm{d}_1 \mathrm{d}_2 \mathscr{F}[\varepsilon_0,0](1,v_0)\right) ={}& \frac{\langle v_0, \frac{2}{\varepsilon_0}(\mathcal{L} + (1+\Delta)^2)v_0 \rangle_{L^2}}{\| v_0 \|^2_{L^2}}\,v_0 =\frac{2}{\varepsilon_0}\frac{\langle v_0, (1+\Delta)^2\,v_0 \rangle_{L^2}}{\| v_0 \|^2_{L^2}}\,v_0 = \frac{2}{\varepsilon_0}\frac{\|(1+\Delta)\,v_0 \|^2_{L^2}}{\| v_0 \|^2_{L^2}}\,v_0,\\
    \end{split}
    \end{equation*}
    and so $\mathcal{P}\left(\mathrm{d}_1 \mathrm{d}_2 \mathscr{F}[\varepsilon_0,0](1,v_0)\right)=0$ if and only if $v_0\,(1+\Delta)v_0 = 0$. For any function $v\in H^1(\mathbb{R}^n; \mathbb{R})$ satisfying $(1+\Delta)v=0$, we note that Pokhozhaev's identity (cf.\ Berestycki and Lions \cite[(2.1)]{Berestycki1983GroundState})
\begin{equation*}
\frac{n-2}{2}\|\nabla v\|_{L^2(\mathbb{R}^n)}^2 = \frac{n}{2}\|v\|_{L^2(\mathbb{R}^n)}^2
\end{equation*}
holds for each $n\in\mathbb{N}$ and thus, for $v_0\in H^4_{\mathrm{e}}$ satisfying $(1+\Delta)v_0=0$ we obtain
\begin{equation*}
\|v_0\|_{L^2}^2 = 0
\end{equation*}
which implies that $v_0 = 0$.
\end{proof}

We now present our existence theorem for the local bifurcation curve.
\begin{thm}[Local bifurcation]\label{thm:Local}
Fix $R>0$, $k\in\mathbb{N}_0$ and $n\in\mathbb{N}$ such that Hypothesis~\ref{hyp:ker} holds for $\varepsilon = \varepsilon_{k,n}$. A bifurcation occurs at $\varepsilon=\varepsilon_{k,n}$; that is, for some $\delta>0$, there exist neighbourhoods $\mathcal{N}\subset \mathbb{R}$ of $\varepsilon_{k,n}$ and $\mathcal{V}\subset H^4_{\mathrm{e}}$ of $0$, and analytic functions $\mathcal{E}:(-\delta,\delta)\to \mathcal{N}$, $v:(-\delta,\delta)\to \mathcal{V}$
such that all nontrivial solutions to \eqref{def:Fu} in $\mathcal{U} := \mathcal{N}\times\mathcal{V}$ lie on the solution branch
    \begin{equation}\label{def:C-loc}
        \mathcal{C}_{\delta} := \left\{ (\mathcal{E}(s), v(s)) : s\in(-\delta,\delta)\right\},
    \end{equation}
with $\mathcal{E}(0)=\varepsilon_{k,n}$, $v(0)=0$ and $v'(0)= u_{k,n}$.
\end{thm}
\begin{proof}
The existence of a (unique) local bifurcation curve follows directly from Theorem~\ref{thm:CR} with $X = H^{4}_{\mathrm{e}}$, $Y = L^2_{\mathrm{e}}$, $\varepsilon_0 = \varepsilon_{k,n}$, $v_0 = u_{k,n}$, and $\mathscr{F}(\varepsilon, u)$ given by \eqref{def:Fu}. Condition $(i)$ follows from Lemma~\ref{lem:Fredholm;L}, $(ii)$ from Lemma~\ref{lem:eig} and Hypothesis~\ref{hyp:ker}, and $(iii)$ from Lemma~\ref{lem:transversal}.
\end{proof}

Having proved the local existence of bifurcating solutions to \eqref{e:SH}, we now present additional results regarding the characterisation of the local solution branch $\mathcal{C}_{\delta}$. We begin by considering the rotational invariance of solutions along each local bifurcation branch; we recall the rotational subgroups $X_{k}^4$ defined in~\eqref{def:Xmk}, referring to Remark~\ref{rmk:Xmk} for the cases $k=0,1$, and prove the following result.
\begin{cor}\label{cor:Xmk}
    Fix $R>0$, $k\in\mathbb{N}_0$ and $n\in\mathbb{N}$ as in Theorem~\ref{thm:Local}. There exists some $0<\delta^*\leq\delta$ such that any solution $(\varepsilon,u)$ to \eqref{def:Fu} along the local bifurcation curve $\mathcal{C}_{\delta^*}\subseteq\mathcal{C}_{\delta}$ lies in the rotational subspace $\mathbb{R}\times X^{4}_{k}$.
\end{cor}
\begin{proof}
    The existence of a local bifurcation curve $\mathcal{C}_{\delta^*}\subseteq \mathbb{R}\times X^{4}_{k}$ follows from respectively replacing $H^{4}_{\mathrm{e}}$ and $L^2_{\mathrm{e}}$ with $X^{4}_{k}$ and $X^{0}_{k}$ in the proof of Theorem~\ref{thm:Local}. Since $X^{4}_{k}$ is a subspace of $H^{4}_{\mathrm{e}}$, we conclude that this bifurcation curve and the curve obtained in Theorem~\ref{thm:Local} are equivalent, by uniqueness in $H^{4}_{\mathrm{e}}$.
\end{proof}

Utilising the analytic properties of the local bifurcation curve, we also apply a supplementary theorem presented in Groves and Horn \cite[Theorem 4.1]{groves2018ferrofluid} to characterise the local bifurcation curve $\mathcal{C}_{\delta}$. We state this result as Theorem~\ref{thm:CR-pitch} in the appendix, which we apply in the following corollary.
\begin{cor}\label{thm:Local-Char}
Fix $R>0$, $k\in\mathbb{N}_0$ and $n\in\mathbb{N}$ as in Theorem~\ref{thm:Local}, so that a local bifurcation occurs at $\varepsilon=\varepsilon_{k,n}$, with Taylor expansions
\begin{equation*}
\varepsilon=\varepsilon_{k,n} + s\,\varepsilon_1 + s^2\,\varepsilon_2 + \mathcal{O}(s^3),\qquad\qquad  u(\mathbf{x}) = s\,u_{k,n}(\mathbf{x}) + s^2\,v_1(\mathbf{x}) + \mathcal{O}(s^3)
\end{equation*}
for $|s|<\delta$, $\varepsilon_1,\varepsilon_2\in\mathbb{R}$, and $v_1\in \ker(\mathcal{Q})$, where $\mathcal{Q}:L^2_{\mathrm{e}}\to L^2_{\mathrm{e}}$ is a projection with $\mathrm{Im}(\mathcal{Q}) = \ker(\mathcal{L})$. 

The coefficient $\varepsilon_1$ is given by
\begin{equation*}
\begin{aligned}
    \varepsilon_1 ={}& - \frac{\varepsilon_{k,n}}{2}\,f''(0)\,\frac{\langle u_{k,n}, u_{k,n}^2\rangle_{L^2}}{\|(1+\Delta)\,u_{k,n}\|^2_{L^2}}.\\
\end{aligned}
\end{equation*}
and the bifurcation is transcritical if $\varepsilon_1\neq0$. If $\varepsilon_1=0$, the coefficient $\varepsilon_2$ is given by
\begin{equation*}
\begin{aligned}
    \varepsilon_2 ={}&  - \frac{\varepsilon_{k,n}}{2}\,f''(0)\,\frac{\langle u_{k,n}, u_{k,n}\,v_1\rangle_{L^2}}{\|(1+\Delta)\,u_{k,n}\|_{L^2}^2} - \frac{\varepsilon_{k,n}}{4}\,\,f'''(0)\,\frac{\| u_{k,n} \|_{L^4}^4}{\|(1+\Delta)\,u_{k,n}\|_{L^2}^2}
\end{aligned}
\end{equation*}
and the bifurcation is a subcritical (supercritical) pitchfork if $\varepsilon_2<0$ ($\varepsilon_2>0$). In particular, $\varepsilon_1=0$ for all $k\neq0$.
\end{cor}
\begin{proof}
The characterisation of the local bifurcation curve follows directly from Theorem~\ref{thm:CR-pitch}, where we note that
for $k\neq0$ we obtain $u_{k,n}^2(r\cos(\theta),r\sin(\theta)) = \frac{1}{2}v_{k,n}^2(r) + \frac{1}{2}v_{k,n}^2(r)\,\cos(2k\theta)$ and so, by the orthogonality of the angular Fourier modes, it follows that
\begin{equation*}
\begin{aligned}
    \varepsilon_1 ={}& - \frac{\varepsilon_{k,n}}{2}\,f''(0)\,\frac{\langle u_{k,n}, u_{k,n}^2\rangle_{L^2}}{\|(1+\Delta)\,u_{k,n}\|^2_{L^2}} = 0\\
\end{aligned}
\end{equation*}
for any $k,n\in\mathbb{N}$.
\end{proof}

We now consider the stability of solutions along the local branch $\mathcal{C}_{\delta}$. In particular, we prove the spectral (in)stability of local bifurcating solutions for transcritical, subcritical, or supercritical bifurcations.

\begin{lem}\label{lem:Stab}
    Fix $R>0$, $k\in\mathbb{N}_0$ and $n\in\mathbb{N}$ as in Theorem~\ref{thm:Local}, so that a local bifurcation occurs at $\varepsilon=\varepsilon_{k,n}$, with bifurcating profile $u(\mathbf{x}) = s\,u_{k,n}(\mathbf{x}) + \mathcal{O}(s^2)$ for $|s|<\delta$. If $\varepsilon=\varepsilon_{k,n}$ is the primary bifurcation point, then the bifurcation curve consists of 
    \begin{enumerate}[label=(\roman*)]
        \item 1 stable and 1 unstable branch if the bifurcation is transcritical;
        \item 2 stable branches if the bifurcation is a supercritical pitchfork; and
        \item 2 unstable branches if the bifurcation is a subcritical pitchfork.
    \end{enumerate}
    Each subsequent local bifurcation curve consists of solely unstable branches.
\end{lem}
    \begin{proof}
We only consider the primary bifurcation branch where the trivial state first destabilises; the eigenvalues of the bifurcating solution will be a small perturbation of the eigenvalues of the trivial state (cf.\ Kato \cite[Theorem VIII-2.6]{Kato1995}), and so each bifurcating solution will be unstable for an unstable trivial state.

We consider the eigenvalue problem for the linearisation of \eqref{e:SH} about the solution $u = v(s)$
\begin{equation}\label{e:SH-stab}
    \lambda w = \mathcal{L}_{v(s)}\,w = -(1+\Delta)^2 w - \varepsilon^2 w + 2\varepsilon^2\,V(|\mathbf{x}|)\,w + f'(v(s))\,w
\end{equation}
where we take $\lambda=\lambda^*$ to be the most unstable eigenvalue; i.e.,
\begin{equation*}
\mathrm{Re}(\lambda^*) > \mathrm{Re}(\lambda)\qquad  \text{for all \;$\lambda\in\sigma(\mathcal{L}_{v(s)})\backslash\{\lambda^*\}$}.    
\end{equation*}
Since we have assumed that the trivial state destabilises at the bifurcation point, it follows that $\lambda^*|_{s=0} = 0$. We note that $\lambda^*|_{s=0} = 0$ is simple, by Hypothesis~\ref{hyp:ker}, and so there exists an analytic function $\Lambda:(-\delta,\delta)\to\mathbb{R}$ such that $\lambda^* = s\,\Lambda(s)$ in a neighbourhood of the curve $\mathcal{C}_{\delta}$ (see Kato \cite[VII-\S3.2]{Kato1995}). We consider the following Taylor expansions
\begin{equation*}
v(s) = u_{k,n}\,s + \sum_{n=1}^{\infty} v_n\,s^{n+1}, \qquad \mathcal{E}(s) = \varepsilon_{k,n} + \sum_{n=1}^{\infty} \varepsilon_n\,s^n, \qquad \Lambda(s) = \lambda_0 + \sum_{n=1}^{\infty} \lambda_n\,s^n,
\end{equation*}
where $v_1, v_2, \dots\in\mathrm{ker}(\mathcal{Q})$, $\varepsilon_1,\varepsilon_2,\dots\in\mathbb{R}$ and $\lambda_0,\lambda_1,\dots\in\mathbb{R}$, such that
\begin{equation*}
 u = v(s),\qquad  \varepsilon = \mathcal{E}(s), \qquad \lambda = s\,\Lambda(s),\qquad s\in(-\delta,\delta)
\end{equation*}
for some $\delta>0$. Expanding \eqref{e:SH-stab} in powers of $s\in(-\delta,\delta)$, we obtain
\begin{equation*}
\begin{aligned}
    \mathcal{O}(1):& &\qquad 0 ={}& -(1+\Delta)^2 w - \varepsilon_{k,n}^2 w + 2\,\varepsilon_{k,n}^2\,V(|\mathbf{x}|)\,w,\\
    \mathcal{O}(s):& &\qquad \lambda_0 w ={}& - 2\,\varepsilon_{k,n}\,\varepsilon_1 w + 4\,\varepsilon_{k,n}\,\varepsilon_1\,V(|\mathbf{x}|)\,w + f''(0)\,u_{k,n}\,w,\\
    \mathcal{O}(s^2): & & \lambda_1 w ={}& -\left(\varepsilon_1^2 + 2 \varepsilon_{k,n} \,\varepsilon_2\right)\,[-1 + 2 V(|\mathbf{x}|)]\,w + f''(0)\,v_1\,w + \frac{1}{2!}\,f'''(0)\,u_{k,n}^2\,w,\\
\end{aligned}
\end{equation*}
The first equation yields $w = u_{k,n}$, resulting in the condition
\begin{equation*}
\begin{aligned}
    \lambda_0 \langle u_{k,n},u_{k,n}\rangle_{L^2} ={}& \varepsilon_1\,\langle u_{k,n}, \mathrm{d}_1\mathrm{d}_2\mathscr{F}[\varepsilon_{k,n},0](1,u_{k,n})\rangle_{L^2} + f''(0)\,\langle u_{k,n}, u_{k,n}^2\rangle_{L^2},\\
\end{aligned}
\end{equation*}
and so
\begin{equation*}
\begin{aligned}
    \lambda_0 ={}& \frac{2\,\varepsilon_1}{\varepsilon_{k,n}}\,\frac{\|(1+\Delta)\,u_{k,n}\|^2_{L^2}}{\|u_{k,n}\|^2_{L^2}} + f''(0)\,\frac{\langle u_{k,n}, u_{k,n}^2\rangle_{L^2}}{\|u_{k,n}\|^2_{L^2}}.\\
\end{aligned}
\end{equation*}
Note the critical case when $\lambda_0 = 0$ is the condition that characterises a transcritical bifurcation in Theorem~\ref{thm:CR-pitch}. For $\lambda_0\neq0$, we see that $\lambda = s\,\Lambda(s)$ changes sign at $s=0$ for $|s|\ll1$, and so one bifurcating branch is stable and the other is unstable. 

If $\langle u_{k,n}, u_{k,n}^2\rangle_{L^2}=0$, we instead set $\varepsilon_1=0$ and obtain the leading-order condition
\begin{equation*}
\begin{aligned}
    \lambda_1 ={}& \frac{2\,\varepsilon_2}{\varepsilon_{k,n}}\frac{\|(1+\Delta)\,u_{k,n}\|_{L^2}^2}{\| u_{k,n} \|_{L^2}^2} + f''(0)\,\frac{\langle u_{k,n}, u_{k,n}\,v_1\rangle_{L^2}}{\| u_{k,n} \|_{L^2}^2} + \frac{1}{2!}\,f'''(0)\,\frac{\| u_{k,n} \|_{L^4}^4}{\| u_{k,n} \|_{L^2}^2}
\end{aligned}
\end{equation*}
where the critical case $\lambda_1 = 0$ is the condition that characterises a pitchfork bifurcation in Theorem~\ref{thm:CR-pitch}. If $f''(0)\,\langle u_{k,n}, u_{k,n}\,v_1\rangle_{L^2} + \frac{1}{2!}\,f'''(0)\,\| u_{k,n}\|_{L^4}^4 > 0$, then the pitchfork bifurcation is subcritical and the bifurcating branches are unstable. Conversely, if $f''(0)\,\langle u_{k,n}, u_{k,n}\,v_1\rangle_{L^2} + \frac{1}{2!}\,f'''(0)\,\| u_{k,n}\|_{L^4}^4 < 0$, then the pitchfork bifurcation is supercritical and the bifurcating branches are stable.
    \end{proof}

\begin{figure}[t!]
    \centering
    \includegraphics[width=\linewidth]{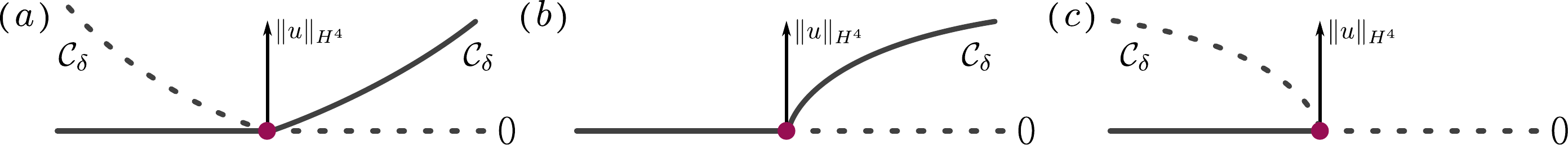}
    \caption{Stability of the primary bifurcating branch for (a) transcritical, (b) supercritical pitchfork, and (c) subcritical pitchfork bifurcations from the trivial state. Here solid and dashed lines indicate stable and unstable solutions, respectively.}
    \label{fig:Stability}
\end{figure}

Thus, we obtain the expected stability behaviour associated with transcritical and pitchfork bifurcations; the axisymmetric solution bifurcates along a stable and unstable branch when $f''(0)\langle u^{\vphantom{2}}_{k,n}, u_{k,n}^2\rangle_{L^2}\neq0$, whereas the non-axisymmetric branch (or the axisymmetric branch with $f''(0)\langle u^{\vphantom{2}}_{k,n}, u_{k,n}^2\rangle_{L^2}=0$) bifurcates along two stable or two unstable branches, depending on the criticality of the pitchfork; see Figure~\ref{fig:Stability}.

\begin{rmk}
We note that Lemma~\ref{lem:Stab} holds when restricted to any rotational subspace $X^4_k$ that $\mathcal{C}_\delta$ lies in. This provides additional information regarding the stability of bifurcating solutions with respect to different dihedral perturbations. 

As an example, suppose the first two bifurcation points are of the form $\varepsilon = \varepsilon_{1,1}, \;\varepsilon_{0,1}$, with each respective solution $u = u_{1,1}, \; u_{0,1}$ bifurcating supercritically. Then, in the full space $H^4_{\mathrm{e}}$ Lemma~\ref{lem:Stab} states that the solution $u = u_{1,1}$ is linearly stable close to its bifurcation point and $u = u_{0,1}$ is unstable. If we instead restrict to the space of axisymmetric functions, then $u= u_{1,1}$ is no longer a valid solution and so $u = u_{0,1}$ is the primary  branch and thus linearly stable by Lemma~\ref{lem:Stab}. It follows that, in this case, $u = u_{0,1}$ is linearly stable with respect to all dihedral and axisymmetric perturbations (i.e. perturbations within the $\mathbb{D}_k$ symmetry group for $k\neq1$) but is unstable in $H^4_{\mathrm{e}}$, while the solution $u = u_{1,1}$ is linearly stable for all perturbations in $H^4_{\mathrm{e}}$.
\end{rmk}

\begin{rmk}
    For particular choices of the nonlinearity $f(u)$, we can immediately predict the bifurcation structure of these localised patterns. In particular, for $f(u)=-u^3$ (or, indeed, any analytic function $f$ with $f''(0)=0$ and $f'''(0)<0$) every bifurcation will be a supercritical pitchfork, and so the primary bifurcating branches will both be stable.
\end{rmk}

We augment our analysis with numerical simulations (using codes available at \cite{Hill2025Github}) of \eqref{e:SH} with $f(u)=\nu u^2 - u^3$, which we present in Figure~\ref{fig:Bif_1}. We consider three cases $R= 2.6, 2.8, 7.5$ for the width of the potential $V$, corresponding to the three regimes highlighted in Figure~\ref{fig:Roots}, and two choices $\nu = 0, 0.4$ of the quadratic coefficient. As predicted, we observe that the primary bifurcation branch is linearly stable to the right of the bifurcation point and unstable to the left, and Figures~\ref{fig:Bif_1}(a-b) and (d-e) show that the primary branch changes from the $\mathbb{D}_0$ spot to the $\mathbb{D}_1$ dipole pattern as $R$ passes from $2.6$ to $2.8$. Furthermore, Figures~\ref{fig:Bif_1}(c) and (f) show that a $\mathbb{D}_6$ branch bifurcates for $R=7.5$, as predicted in Figure~\ref{fig:Roots}(d). 

We note that, by including the quadratic coefficient $\nu=0.4$, the $\mathbb{D}_0$ branch changes from a supercritical pitchfork to a transcritical bifurcation, while the $\mathbb{D}_1$ branch changes from a supercritical to a subcritical pitchfork when $R=2.8$; in contrast, the $\mathbb{D}_1$ branch remains unchanged for $R=2.6$, as does the $\mathbb{D}_6$ branch for $R=7.5$. While the $\mathbb{D}_1$ branch is qualitatively the same in Figures~\ref{fig:Bif_1}(a) and (d), and likewise the $\mathbb{D}_6$ branch in Figures~\ref{fig:Bif_1}(c) and (f), we emphasise that the solution profiles qualitatively change as $\nu$ increase from $0$ to $0.4$. We note that when $R=2.8$, such that the primary branch is the $\mathbb{D}_1$ dipole pattern, the $\mathbb{D}_0$ branch is initially unstable but then restabilises as a pair of unstable $\mathbb{D}_1$ branches emerge from a secondary bifurcation; see Figures~\ref{fig:Bif_1}(b) and (e).

\begin{figure}[t!]
    \centering
    \includegraphics[width=\linewidth]{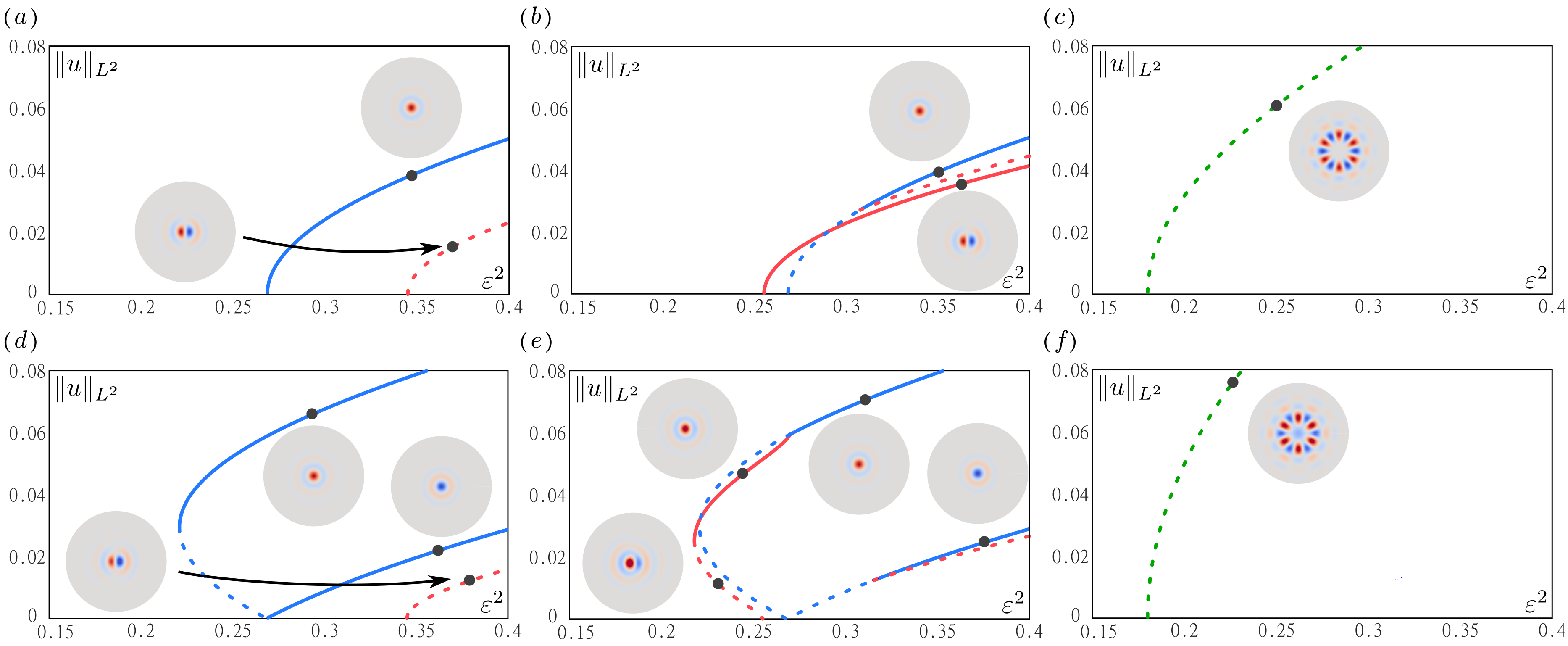}
    \caption{Numerical simulations of \eqref{e:SH} with $f(u) = \nu\, u^2 -u^3$. Columns indicate (Left) $R = 2.6$, (Centre) $R=2.8$, and (Right) $R=7.5$, while rows indicate (Top) $\nu=0$ and (Bottom) $\nu=0.4$. The blue, red and green curves are the solution branches for localised $\mathbb{D}_0$, $\mathbb{D}_1$ and $\mathbb{D}_6$ patterns, respectively, and solid/dashed lines indicate linearly stable/unstable solutions.}
    \label{fig:Bif_1}
\end{figure}   

\subsection{Global bifurcation}\label{s:bif-glob}
We now prove that the local bifurcation curves found in the previous section can be continued away from the bifurcation point; we call these \textit{global bifurcation branches}, and use the analytic global implicit function theorem of Chen, Walsh and Wheeler \cite[Theorem B.1]{Chen2024Global} (which we present as Theorem~\ref{thm:Global-cont} in the appendix) in order to prove their existence.

\begin{thm}[Global bifurcation]\label{thm:Global}
Fix $R>0$, $k\in\mathbb{N}_0$ and $n\in\mathbb{N}$ as in Theorem~\ref{thm:Local}, so that there exists a local bifurcation curve $\mathcal{C}_{\delta}$ passing through the point $(\varepsilon,u) = (\varepsilon_{k,n},0)$. Then, there exists an open set $U\subset \mathbb{R}\times H^4_{\mathrm{e}}$, with $(\varepsilon_{k,n},0)\in U$, such that $\mathcal{C}_{\delta}$ can be extended to a global bifurcation curve $\mathcal{C} = \{ (\mathcal{E}(s), v(s))\,:\; s\in\mathbb{R}\}$ of solutions to \eqref{def:Fu}, where $(\mathcal{E}, v) :\mathbb{R}\to U$ is a continuous function. Furthermore,
\begin{enumerate}[label=(\alph*)]
    \item $\mathcal{C}$ has a local analytic parametrisation around each of its points;
     \item any locally analytic bifurcation curve in $U$ which includes the point $(\varepsilon_{k,n},0)$ is contained in $\mathcal{C}$; and
     \item one of the following alternatives occurs:
     \begin{enumerate}[label=\Roman*:~]
         \item $\mathcal{C}$ is unbounded;
         \item $\mathcal{C}$ approaches the boundary of $U$;
         \item $\mathcal{C}$ forms a closed loop; or
         \item $\mathcal{C}$ loses compactness.
     \end{enumerate}
 \end{enumerate}
\end{thm}
\begin{proof}
This result follows directly from Theorem~\ref{thm:Global-cont} with $X = H^{4}_{\mathrm{e}}$ and $Y = L^{2}_{\mathrm{e}}$. Upon choosing some $(\varepsilon_*,u_*)\in\mathcal{C}_\delta$ such that the linearisation of \eqref{e:SH-steady} at $(\varepsilon_*,u_*)$ is an isomorphism, one can apply the analytic implicit function theorem and use theory regarding analytic curves to globally extend the local solution branch; see, for example, the proofs of Theorem 9.1.1 in \cite{BuffoniToland2003Global}, Theorem 6.1 in \cite{Chen2017Solitary}, or Theorem 2.4 in \cite{Doak2026Piecewise}. One typically continues the local branch $\mathcal{C}_\delta$ from a starting point $(\varepsilon_*,u_*)\in\mathcal{C}_\delta$ which is distinct from the point $(\varepsilon_0,0)$; this is because either $(\varepsilon_0,0)$ is a bifurcation point and thus the linear operator is not invertible at $\varepsilon=\varepsilon_0$, such as in \cite[Proof of Theorem 9.1.1]{BuffoniToland2003Global}, or $(\varepsilon_0,0)$ lies on the boundary of the set $U$, such as in \cite[Proof of Theorem 6.1]{Chen2017Solitary}. 
 
By Lemma~\ref{lem:Fredholm;L}, $\mathcal{L}_{\phi}:=\mathrm{d}_2\mathscr{F}[\varepsilon,\phi]$ is a Fredholm operator with index zero along the local bifurcation branch $\mathcal{C}_\delta$ and by Lemma~\ref{lem:transversal} an isolated eigenvalue of $\mathcal{L}_{\phi}$ passes through zero at $(\varepsilon, \phi) = (\varepsilon_0, 0)$. It follows that there is a non-empty subset $\mathcal{C}_{*}\subseteq\mathcal{C}_{\delta}$ along which $\mathcal{L}_{\phi}$ is an isomorphism; we choose some point $(\varepsilon_*, u_*)\in\mathcal{C}_{*}$ and apply Theorem~\ref{thm:Global-cont}.

\end{proof}

We make the following remarks regarding Theorem~\ref{thm:Global}.
\begin{rmk}\label{rmk:global} \,

\begin{enumerate}[label = (\roman*)]
    \item Much like in Corollary~\ref{cor:Xmk}, Theorem~\ref{thm:Global} holds when replacing the respective spaces $H^{4}_{\mathrm{e}}$ and $L^2_{\mathrm{e}}$ with $X^4_{k}$ and $X^0_k$, and so the global bifurcation curve is also contained in the rotational subspace $\mathbb{R}\times X^4_k$.
    \item We note that $\mathcal{C}$ having an analytic parametrisation at each point does not imply that $\mathcal{C}$ is smooth. In fact, as discussed in \cite[Remark 9.1.2 (3)]{BuffoniToland2003Global}, the curve $\mathcal{C}$ can possess isolated points where $\mathcal{C}$ is not even continuously differentiable. We demonstrate an example where $\mathcal{C}$ has an isolated corner in Figure~\ref{fig:Bif_2}(iv).
    \item Different versions of the alternatives I-III are presented in the literature (see, for example \cite[Remarks 4.2 and 4.3]{boehmer2025patternformationfilmrupture}); in particular, Constantin, Strauss and V\u{a}rv\u{a}ruc\u{a}~\cite[Remark 8]{Constantin2016} highlight that the alternatives presented by Buffoni and Toland~\cite[Theorem 9.1.1]{BuffoniToland2003Global} are not exhaustive, and present a different formulation in their global bifurcation theorem~\cite[Theorem 6]{Constantin2016}. Here we use the same formulation of the alternatives as used by Chen, Walsh and Wheeler~\cite[Theorem B.1]{Chen2024Global}, which we present in more detail in Theorem~\ref{thm:Global-cont}. Note that alternatives $I$ and $II$ are often indistinguishable from one another and so we combine the two into alternative $(i)$ in Theorem~\ref{thm:Global-cont}.
    \item In many applications one can exclude the case when $\mathcal{C}$ forms a closed loop in part $(c)$, typically via bifurcation theory on cones \cite[Theorem 9.2.2]{BuffoniToland2003Global} and using nodal properties of the bifurcating solution. However, in our problem we numerically observe cases where $\mathcal{C}$ forms a closed loop (such as in Figure~\ref{fig:Bif_2}(iv)) and so we do not expect this to be possible in general.
    \item The loss of compactness alternative can sometimes be excluded through certain properties of solutions; for example, Haziot and Wheeler~\cite{Haziot2023ConstantVorticity} used the monotonicity of their solutions while Kozlov~\cite{kozlov2025nonuniqueness} utilised uniform decay rates. For particular choices of nonlinear function $f$ where one can prove that there are no non-trivial solutions to the homogeneous problem (i.e.\ with $V\equiv0$), such as $f(u)=-u^3$, one might be able to apply a concentration-compactness argument~\cite{Lions1984a,Lions1984b} to eliminate the loss of compactness. We have been unable to exclude a loss of compactness in the current setting outside of axisymmetric patterns---which we present in the following lemma and discuss further in Section~\ref{app:comp} of the Appendix---but it may be possible to further strengthen our global bifurcation result in future studies.
\end{enumerate}
\end{rmk}

In the case of axisymmetric solutions $u\in X^4_0$, we obtain the following pre-compactness result.
\begin{lem}\label{lem:Compact}
    Let $\mathcal{C}\subset X^4_{0}$ be a global bifurcation curve for a localised axisymmetric solution to \eqref{e:SH-steady}. Then, alternative IV in Theorem~\ref{thm:Global} cannot occur.
\end{lem}
We prove a more general version of this result (Lemma~\ref{lem:Compact-0}) in the Appendix, where we rely on an \textit{a priori} uniform decay rate for radial functions. Such a uniform decay rate is not known for non-radial functions, and we would not expect there to be an analogous result for all functions due to translational invariance of the space $H^4(\mathbb{R}^2)$; for further discussion of alternative $IV$, see Section~\ref{app:comp}.

We again provide numerical simulations to support our analysis, which we present in Figure~\ref{fig:Bif_2}. For fixed $R=2.8$, we can vary $\nu$ such that we obtain the three different alternatives in Theorem~\ref{thm:Global} for the primary $\mathbb{D}_1$ bifurcation. As such, we are unable to disregard any of the possible alternatives presented in Theorem~\ref{thm:Global}.

For $\nu\neq0$, the dipole pattern evolves into a combination of a radial spot and the dipole; for $\nu=0.4$ this results in a loop connecting the bifurcation point to a secondary bifurcation at the radial spot, while for $\nu=1.6$ the branch converges towards $\varepsilon=0$. In this case the solution returns to a dipole solution but with weakening localisation, as plotted in Figure~\ref{fig:Bif_2} $(vi)$. We note that each branch plotted in Figure~\ref{fig:Bif_2} is actually two branches, the solution and its half-period rotation, which is why the closed loop appears to be a single branch terminating at a point. We also note that the point $(iv)$ in Figure~\ref{fig:Bif_2} forms an isolated corner on the closed loop, which is permitted within Theorem~\ref{thm:Global} (see Remark~\ref{rmk:global}$(iii)$). 

We briefly comment on the boundary $\varepsilon=0$; this corresponds to the parameter values at which the essential spectrum of the linear operator $\mathcal{L}$ touches the imaginary axis, known as the onset of \textit{essential instability} \cite{Sandstede1999Essential-waves}. Essential instabilities are typical of homogeneous pattern-forming systems, such as the Swift-Hohenberg equation \eqref{e:SH} with $V\equiv0$, where Crandall-Rabinowitz bifurcation theory can no longer be applied due to the loss of Fredholm properties at $\varepsilon=0$. There have been several studies of localised states at the onset of essential instability in one spatial dimension, using techniques from spatial dynamics \cite{Sandstede2001Essential-Bif,Sandstede2000Essential-stability}, but extending such analysis to higher spatial dimensions remains an open problem.

\begin{figure}[t!]
    \centering
    \includegraphics[width=\linewidth]{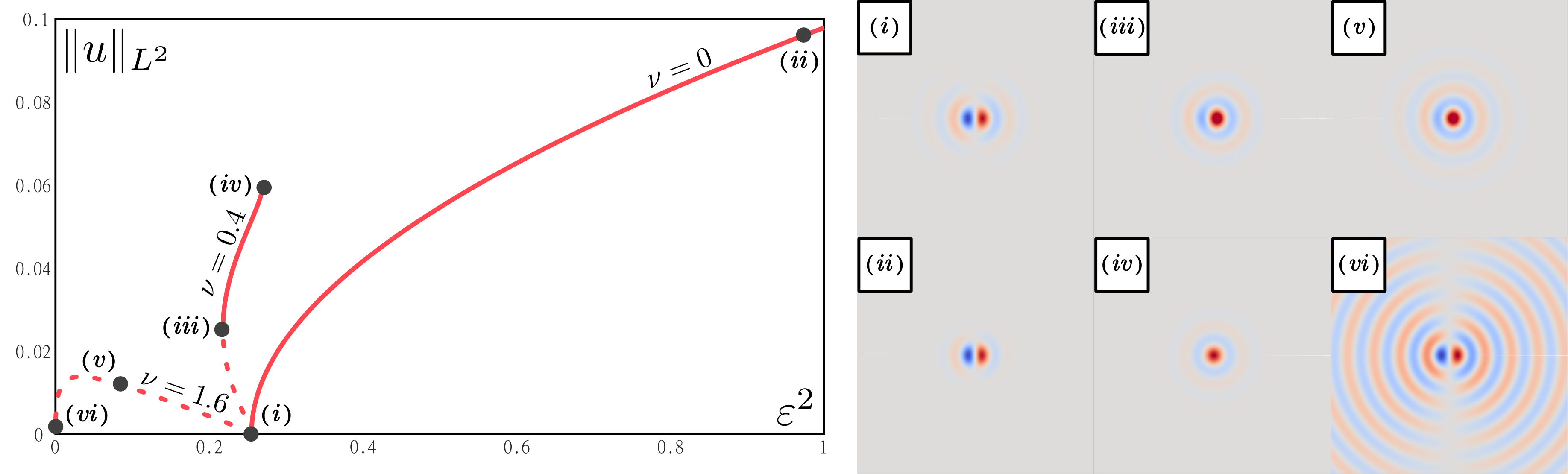}
    \caption{Numerical simulation of the primary branch, consisting of a $\mathbb{D}_1$ dipole pattern, for \eqref{e:SH} with $f(u) = \nu u^2 - u^3$, $R=2.8$, and $\nu = 0, 0.4, 1.6$. Line style and colour of branches indicates the same properties as in Figure~\ref{fig:Bif_1}. The solution $(i\to ii)$ blows up in its norm for $\nu=0$, $(i \to iii\to iv)$ forms a closed loop for $\nu=0.4$, and $(i\to v \to vi)$ approaches the boundary $\varepsilon=0$ for $\nu=1.6$.}
    \label{fig:Bif_2}
\end{figure}

\section{Discussion}\label{s:discuss}

In this paper we have proved the existence of large-amplitude, fully localised 2D patterns in the Swift--Hohenberg equation with an axisymmetric spatial heterogeneity. It is somewhat remarkable that one is able to do this despite there being no equivalent results in the spatially homogeneous case; indeed, the spatial heterogeneity considered in this paper allows us to derive almost explicit characterisations of the shape of the 2D localised pattern and its linear stability near bifurcation. Despite choosing an axisymmetric heterogeneity, non-axisymmetric localised solutions can emerge and even bifurcate as a stable localised pattern. This raises the question of whether such patterns can be realised experimentally. 

This work lays the foundation for further exploration of localised patterns induced by compact spatial heterogeneities. A straightforward extension would be to consider the three-dimensional analogue of \eqref{e:SH}, where the Fredholm properties will remain the same and the point spectrum can again be computed explicitly using spherical harmonics and spherical Bessel functions. Replacing the step function $V$ by a smooth potential would be a more realistic model for spatial heterogeneities in experiments, and one would expect much of the underlying theory of this work to persist in that case. In particular, we highlight the approach presented by Brooks, Derks and Lloyd~\cite[Section 5]{Brooks_2019} in which they prove the persistence of fronts in the sine-Gordon equation from piecewise to smooth heterogeneities; we expect a similar argument to hold in this work using the approach of van Heijster and Sandstede~\cite{Heijster_2011}. 

Different types of spatial heterogeneity can be considered, for instance, non-axisymmetric heterogeneities like those studied in~\cite{Jaramillo2015inhom-stripes,Jaramillo2023inhom-targets} or an axisymmetric heterogeneity translating uniformly in space. The most challenging part in these cases is then to explicitly compute the point spectrum of the linear operator. Another common type of spatial heterogeneity is when $u$ is no longer multiplying $V$, e.g.
\begin{equation*}
    \partial_t u = -(1+\Delta)^2 u - \varepsilon^2 u + f(u) + \delta\,V(\mathbf{x}).
\end{equation*}
This type of heterogeneity is more similar to the form of heterogeneity found in the Rayleigh--B\'{e}nard convection problem, where a non-uniform heat profile is applied to the bottom of the fluid. This problem has the additional difficulty that the trivial state is no longer an equilibrium state, but the small defect case (i.e. with $|\delta|\ll1$) may still be amenable to methods presented here. 

We recall that our approach is inspired by experimental and numerical studies, where a spatial heterogeneity is added to the system in order to induce the emergence of localised patterns. However, there is a key element of this approach which we have not considered here, namely, the persistence of these patterns. In the ferrofluid experiment, for example, the spatial heterogeneity is subsequently removed from the system and the induced localised spikes remain on the fluid surface. This represents a method by which one can find localised patterns in a spatially homogeneous system by adding and removing a compact heterogeneity, which we would like to replicate in our analysis. Mathematically, this is equivalent to finding $(\mathcal{E}(s),v(s))\in \mathcal{C}$ on the global bifurcation curve for \eqref{e:SH}, and then seeking solutions to the homogeneous Swift--Hohenberg equation (i.e. \eqref{e:SH} with $V\equiv0$) of the form $\varepsilon=\mathcal{E}(s)$, $u = v(s) + w$ with $w\in H^4_{\mathrm{e}}$. In particular, one would need to solve
\begin{equation*}
    [(1+\Delta)^2 + \mathcal{E}^2 - f'(v)] w = - 2\,\mathcal{E}^2\,V\,v + N(w;v)
\end{equation*}
for $w\in H^4_{\mathrm{e}}$, where $N(w;v) := f(v+w) - f(v) - f'(v)\,w = \mathcal{O}(|w|^2)$, possibly via a Newton--Kantorovich or implicit function argument. Numerically, we observe that the remainder function $w$ is not small relative to $v$, making it difficult to rule out the trivial solution $w\equiv-v$ (corresponding to $u\equiv0$). It may still be possible to prove the existence of a non-trivial solution---thereby solving the open problem of proving the existence of fully localised 2D patterns---but we leave this as a future avenue to explore. 

We expect the approach presented in this work to be easily extended to other types of PDE systems, and so it would be natural to consider models that are more specific to particular experiments or phenomena. As discussed, the ferrofluid experiment studied by Richter and Barashenkov~\cite{Richter2005ferro-spikes} is a key motivation for this work, and so it would be interesting to try and apply this theory to the ferrofluid problem. We note that there have already been studies of two-dimensional neural field models with spatial heterogeneity~\cite{Rankin2014}, including the case where firing rate is piecewise smooth~\cite{Avitabile2015,Bressloff2011,Owen_2007}, and so this may be a useful application to explore in greater detail.

\subsection*{Data availability statement}
The data that support the findings of this study are openly available at \cite{Hill2025Github}.

\subsection*{Acknowledgments}
The authors are grateful to Bj\"orn de Rijk, Mark D. Groves and Miles Wheeler for their assistance on certain aspects of this manuscript. DJH acknowledges support from the Alexander von Humboldt Foundation, and MT and DJBL acknowledge support from EPRSC grant UKRI070. The authors also thank Tom Bridges, Ryan Goh, Bastian Hilder, J\"org Weber, Gui-Qiang Chen, Melanie Rupflin and Sergey Zelik for helpful discussions and suggestions regarding this work. For the purpose of Open Access, the authors have applied a Creative Commons Attribution (CC BY) public copyright licence to any Author Accepted Manuscript version arising from this submission.

\appendix
\setcounter{equation}{0}
\renewcommand\theequation{\Alph{section}.\arabic{equation}}

\section{Bifurcation theorems}\label{app:bif}

There are several bifurcation theorems that we have used in this work, each of which is a standard tool in bifurcation theory. To balance accessibility and brevity, we present the full statement of each theorem that we use and provide further references for those interested in their proofs.   

In our local bifurcation theorem (Theorem~\ref{thm:Local}) we use an analytic version of the Crandall--Rabinowitz theorem (cf.\ Buffoni and Toland \cite[Theorem 8.3.1]{BuffoniToland2003Global}), which we state here with the notation presented by Groves and Horn \cite[Theorem 3.1]{groves2018ferrofluid}.

\begin{thm}[Crandall--Rabinowitz theorem]\label{thm:CR}
Let $X$ and $Y$ be Banach spaces, $V$ be an open neighbourhood of the origin in $X$ and $\mathscr{F}: \mathbb{R}\times V \to Y$ be an analytic function with $\mathscr{F}(\varepsilon, 0) = 0$ for all $\varepsilon\in\mathbb{R}$. Suppose also that
\begin{enumerate}[label=(\roman*)]
    \item $\mathcal{L}:= \mathrm{d}_2\mathscr{F}[\varepsilon_0,0] : X \to Y$ is a Fredholm operator of index zero,
    \item $\mathrm{ker}(\mathcal{L})= \mathrm{span}\{ v_0 \}$ for some $v_0\in X$, and
    \item the transversality condition $\mathcal{P}(\mathrm{d}_1 \mathrm{d}_2 \mathscr{F}[\varepsilon_0,0](1,v_0)) \neq 0$ holds, where $\mathcal{P} : Y \to Y$ is a projection with $\mathrm{Im}(\mathcal{L})= \mathrm{ker}(\mathcal{P}).$
\end{enumerate}
 The point $(\varepsilon_0,0)$ is a local bifurcation point, that is there exist $\delta>0$, an open neighbourhood $W$ of $(\varepsilon_0,0)$ in $\mathbb{R}\times X$ and analytic functions $v: (-\delta,\delta)\to V$, $\mathcal{E} : (-\delta,\delta)\to \mathbb{R}$ with $\mathcal{E}(0) = \varepsilon_0$, $v(0) = 0$, $v'(0) = v_0$ such that $\mathscr{F}(\mathcal{E}(s), v(s)) = 0$ for every $s\in(-\delta,\delta)$. Furthermore, 
 \begin{equation*}
W\cap N = \{(\mathcal{E}(s), v(s)) : 0 <|s|<\delta\}, 
 \end{equation*}
where 
\begin{equation*}
N= \{(\varepsilon,u) \in\mathbb{R}\times  V \backslash \{0\} : \mathscr{F}(\varepsilon, u) = 0\}.
\end{equation*}
\end{thm}

A proof of the standard Crandall--Rabinowitz theorem can be found in the original paper by Crandall and Rabinowitz~\cite[Theorem 1]{Crandall1971Bifurcation}, while the analytic version is proved by Buffoni and Toland~\cite[Theorem 8.3.1]{BuffoniToland2003Global}. Using the analytic properties of the local bifurcation curve, we can also apply the following supplementary theorem presented by Groves and Horn \cite[Theorem 4.1]{groves2018ferrofluid} to characterise the local bifurcation curve $\mathcal{C}_{\delta}$.
\begin{thm}\label{thm:CR-pitch}
Suppose that the hypotheses of Theorem \ref{thm:CR} hold. In the notation of that theorem, let $\mathcal{Q}:X\to X$ be a projection with $\mathrm{Im}(\mathcal{Q}) = \mathrm{ker}(\mathcal{L})$ and the Taylor series of the functions $v:(-\delta, \delta)\to V$, $\mathcal{E}:(-\delta,\delta)\to\mathbb{R}$ be given by 
\begin{equation*}
v(s) = v_0\,s + \sum_{n=1}^{\infty} v_n\,s^{n+1}, \qquad \mathcal{E}(s) = \varepsilon_0 + \sum_{n=1}^{\infty} \varepsilon_n\,s^n,
\end{equation*}
where $v_1, v_2, \dots\in\mathrm{ker}(\mathcal{Q})$ and $\varepsilon_1,\varepsilon_2,\dots\in\mathbb{R}$.
\begin{enumerate}[label=(\roman*)]
\item The coefficient $\varepsilon_1$ satisfies the equation
\begin{equation*}
\mathcal{P}\left(\frac{1}{2!}\,\mathrm{d}_{2}^2\mathscr{F}[\varepsilon_0,0](v_0,v_0)\right) + \varepsilon_1\mathcal{P}\left(\mathrm{d}_1\mathrm{d}_2\mathscr{F}[\varepsilon_0,0](1,v_0)\right) = 0,
\end{equation*}
and the bifurcation is transcritical if $\varepsilon_1$ is non-zero.
\item Suppose that $\varepsilon_1=0$. The coefficient $\varepsilon_2$ satisfies the equation
\begin{equation*}
\mathcal{P}\left(\mathrm{d}_{2}^2\mathscr{F}[\varepsilon_0,0](v_0,v_1) + \frac{1}{3!}\,\mathrm{d}_{2}^3\mathscr{F}[\varepsilon_0,0](v_0,v_0,v_0)\right) + \varepsilon_2\mathcal{P}\left(\mathrm{d}_1\mathrm{d}_2\mathscr{F}[\varepsilon_0,0](1,v_0)\right) = 0,
\end{equation*}
where $v_1\in\mathrm{ker}(\mathcal{Q})$ solves the equation 
\begin{equation*}
\mathrm{d}_{2}\mathscr{F}[\varepsilon_0,0](v_1) = -\frac{1}{2!}\,\mathrm{d}_{2}^2\mathscr{F}[\varepsilon_0,0](v_0,v_0)
\end{equation*}
The bifurcation is supercritical for $\varepsilon_2>0$ and subcritical for $\varepsilon_2<0$.
\end{enumerate}
\end{thm}
This theorem uses the analytic properties of bifurcating solutions in order to take Taylor series expansions for $|s|\ll1$ and characterise the bifurcation structure by standard normal forms in bifurcation theory (see, for example Verhulst~\cite[Section 13.5]{verhulst1996nonlinear}).

To prove our global bifurcation result (Theorem~\ref{thm:Global}) we use the analytic global implicit function theorem given by Chen, Walsh and Wheeler~\cite[Theorem B.1]{Chen2024Global}, which extends the global bifurcation theorem of Buffoni and Toland~\cite[Theorem 9.1.1]{BuffoniToland2003Global} to cases where compactness and Fredholm properties of solutions are no longer assumed.
\begin{thm}[Global continuation]\label{thm:Global-cont}
Let $X$ and $Y$ be Banach spaces, $U\subseteq \mathbb{R}\times X$ an open set containing a point $(\varepsilon_*, u_*)$. Suppose that $\mathscr{F}: U\to Y$ is real analytic and satisfies
\begin{equation*}
    \mathscr{F}(\varepsilon_*, u_*)=0, \qquad\qquad \mathrm{d}_2\mathscr{F}[\varepsilon_*, u_*]\;\text{is an isomorphism $X\to Y$.}
\end{equation*}
Then there exists a curve $\mathcal{C}$ that admits the global $C^0$ parametrisation
\begin{equation*}
    \mathcal{C}:=\left\{ (\mathcal{E}(s), v(s))\;:\;\; s\in\mathbb{R}\right\}\subseteq \mathscr{F}^{-1}(0)\cap U,
\end{equation*}
and satisfies the following.
 \begin{enumerate}[label=(\alph*)]
     \item At each $s\in\mathbb{R}$, the linear operator $\mathrm{d}_{2}\mathscr{F}[\mathcal{E}(s), v(s)]:X\to Y$ is Fredholm with index zero.
     \item One of the following alternatives holds as $s\to \infty$ and $s\to-\infty$.
     \begin{enumerate}[label=(\roman*)]
         \item (Blow up) The quantity
         \begin{equation*}
             N(s) := \| v(s)\|_{X} + |\mathcal{E}(s)| + \frac{1}{\mathrm{dist}((\mathcal{E}(s),v(s)), \partial U)} \to \infty
         \end{equation*}
         \item (Loss of compactness) There exists a sequence $\{s_n\}$ with $s_n\to\pm\infty$ and $\sup_n N(s_n)<\infty$, but $(\mathcal{E}(s_n), v(s_n))$ has no convergent subsequence in $\mathbb{R}\times X$.
         \item (Loss of Fredholmness) There exists a sequence $\{s_n\}$ with $s_n\to\pm\infty$, $\sup_n N(s_n)<\infty$, and $(\mathcal{E}(s_n), v(s_n))\to (\mathcal{E}_0,v_0)\in U$ in $\mathbb{R}\times X$, however $\mathrm{d}_2\mathscr{F}[\mathcal{E}_0, v_0]$ is not Fredholm with index zero.
         \item (Closed loop) There exists $T>0$ such that $(\mathcal{E}(s+T), v(s+T)) = (\mathcal{E}(s), v(s))$ for all $s > 0$.
     \end{enumerate}
     \item Near each point $(\mathcal{E}(s_0), v(s_0))\in\mathcal{C}$, we can locally reparametrise $\mathcal{C}$ so that $s\mapsto (\mathcal{E}(s),v(s))$ is real analytic.
     \item The curve $\mathcal{C}$ is maximal in the sense that, if $\mathcal{R}\subseteq \mathscr{F}^{-1}(0)\cap U$ is a locally real-analytic curve containing $(\varepsilon_*, u_*)$ and along which $\mathrm{d}_2\mathscr{F}$ is Fredholm of index zero, then $\mathcal{R}\subseteq\mathcal{C}$.
 \end{enumerate}
\end{thm} 
For a proof of this theorem, we refer the reader to Chen, Walsh and Wheeler~\cite[Theorem B.1]{Chen2024Global}, along with Buffoni and Toland~\cite[Theorem 9.1.1]{BuffoniToland2003Global} for additional details.

\section{Additional proofs}\label{app:spec}

We present several proofs regarding the spectral properties of the linear operator $\mathcal{L}_{\phi}$ given in \eqref{def:L-phi}.

We begin with a proof that multiplication by an element in $H^m=H^m(\mathbb{R}^2)$ is a compact operator from $H^m$ to $L^2=L^2(\mathbb{R}^2)$ with $m>1$; while this result is often used in linear stability analysis (cf.\ Kapitula and Promislow \cite[Theorem 3.1.11]{kapitula2013spectral}), we present a proof here for completeness. 
\begin{lem}\label{lem:compact}
Suppose that $f\in H^m$ with $m>1$. The mapping $u \mapsto f\,u$ is a compact operator $H^m \to L^2$.
\end{lem}
\begin{proof}
Let $\{u_n\}$ be bounded in $H^m$, so that $\{f\, u_n\}$ is weakly convergent in $H^m$. Denote its weak limit by $g$ and note that $f\,u_n \to g$ (up to subsequences) in $L^2(B_{C})$ for any $C>0$, where $B_C$ denotes the unit ball of radius $C$ about the origin. Then, we obtain
\begin{equation*}
\begin{split}
\| f\,u_n - g \|_{L^2}^2 ={}& \int_{|\mathbf{x}|\leq C} |f(\mathbf{x})\,u_n(\mathbf{x}) - g(\mathbf{x})|^2\,\mathrm{d}\mathbf{x} + \int_{|\mathbf{x}|> C} |f(\mathbf{x})\,u_n(\mathbf{x}) - g(\mathbf{x})|^2\,\mathrm{d}\mathbf{x},\\
\leq{}& \| f\,u_n - g \|_{L^2(B_C)}^2 + \| u_n\|_{\infty}^2\,\| f\|_{L^2(\mathbb{R}^2\backslash B_C)}^2 + \| g \|_{L^2(\mathbb{R}^2\backslash B_C)}^2,\\
\leq{}& \varepsilon^2,\\
\end{split}
\end{equation*}
for sufficiently large $n$. First choose $C$ sufficiently large such that the second and third terms are each less than $\frac{\varepsilon^2}{4}$, using the fact that $\| u_n\|_{\infty}$ is uniformly bounded since $H^m$ is continuously embedded in $L^\infty$. It follows that there exists some constant $N = N(C)>0$ such that the first term is less than $\frac{\varepsilon^2}{2}$ for all $n>N$, since $f\,u_n \to g$ in $L^2(B_{C})$.
\end{proof}


Additionally, we     prove Lemma~\ref{lem:eig}, characterising the eigenvalues and eigenfunctions of the linear eigenvalue problem \eqref{e:Stab}. For the sake of readability, we first restate the lemma as originally presented in Section~\ref{s:spec-point}.

\begin{replem}{lem:eig}
    Fix $R>0$ and $\varepsilon>0$. The eigenvalue problem \eqref{e:Stab} possesses a solution $u = u_k \in H^{4}_{\mathrm{e}}$ given by $u_k(r\cos\theta,r\sin\theta) = v_k(r)\,\cos(k\theta)$ with
\begin{equation*}
v_k(r) = \begin{cases}
\displaystyle \cos(\phi + \psi)\frac{J_{k}(\alpha_{+}\,r)}{J_{k}(\alpha_{+}\,R)} + \cos(\phi - \psi)\frac{J_{k}(\alpha_{-}\,r)}{J_{k}(\alpha_{-}\,R)}, & r<R,\\
\displaystyle \sqrt{2}\,\sqrt{1 - \tilde{\lambda}}\, \mathrm{Re}\left(\mathrm{e}^{\mathrm{i}\phi}\, \frac{H^{(1)}_{k}( \beta r)}{H^{(1)}_{k}(\beta R)}\right), & r>R,\\
\end{cases}  
\end{equation*}
where $\alpha_{\pm} := \sqrt{1 \pm \varepsilon\,\sqrt{1 - \tilde{\lambda}}}$, $\beta := \sqrt{1 + \mathrm{i}\,\varepsilon\,\sqrt{1 + \tilde{\lambda}}}$; $J_k, H_k^{(1)}$ are the $k$'th order Bessel functions of the first and third kind; $\psi$ is given by $\psi = \frac{1}{2}\sin^{-1}(\tilde{\lambda}) + \frac{\pi}{4}$; $\phi$ is given by
    \begin{equation*}
\mathrm{e}^{2\mathrm{i}\phi } = \mathrm{i}\,\mathrm{e}^{-\mathrm{i}\sin^{-1}(\tilde{\lambda})}\frac{W[J_{k}(\alpha_{+}r), \overline{H^{(1)}_{k}( \beta r)}](R)\,H^{(1)}_{k}( \beta R)}{W[J_{k}(\alpha_{+}r), H^{(1)}_{k}( \beta r)](R)\,\overline{H^{(1)}_{k}( \beta R)}};
    \end{equation*}
    $W[\cdot,\cdot]$ denotes the weighted Wronskian function given by
\begin{equation*}
    W[u,v](r) := r\left( u(r)\,v'(r) - u'(r)\,v(r)\right);
\end{equation*}
    and $\tilde{\lambda} = \tilde{\lambda}(R, \varepsilon)$ satisfies
\begin{equation*}
F_k(R,\varepsilon,\tilde{\lambda}) = \mathrm{Re}\left[\mathrm{e}^{\mathrm{i}\sin^{-1}(\tilde{\lambda})}\, W[J_{k}(\alpha_{+}r), H^{(1)}_{k}( \beta r)](R)\,W[J_{k}(\alpha_{-}r), \overline{H^{(1)}_{k}( \beta r)}](R)\right]=0 
\end{equation*}
with $\tilde{\lambda}\notin(-\infty,-1]$.
\end{replem}
\begin{proof}[Proof of Lemma~\ref{lem:eig}]
We first decompose $u$ into an angular Fourier expansion
\begin{equation*}
u(r\cos\theta,r\sin\theta) = u_0(r) + \sum_{k=1}^{\infty} u_{k}(r)\,\cos(k\theta),
\end{equation*}
with $u_k\in H^4_{(k)}((0,\infty))$ as defined by Groves and Hill~\cite{groves2024function} for each $k\in\mathbb{N}_0$. The linear problem \eqref{e:Stab} then becomes
\begin{equation*}
\tilde{\lambda}\,\varepsilon^2 u_k = -(1 + \mathcal{D}_{1-k}\mathcal{D}_{k})^2 u_{k} - \varepsilon^2 u_k + 2\,\varepsilon^2\,V(r)\,u_k, \qquad\qquad \forall\; k\in\mathbb{N}_{0},
\end{equation*}
where we use the Bessel differential operators $\mathcal{D}_{k} := \partial_{r} + \frac{k}{r}$ also introduced in \cite{groves2024function}. For notational simplicity, we also define the $k$-index Laplace operator as $\Delta_k := \mathcal{D}_{1-k}\mathcal{D}_{k} = \partial_r^2 + \frac{1}{r}\partial_r - (\frac{k}{r})^2$. We seek solutions to \eqref{e:Stab} for an arbitrary $k\in\mathbb{N}_{0}$; to do this, we construct real bounded solutions on the disjoint regions $r<R$ and $r>R$ and match these functions at the discontinuity point $r=R$.

In both regions $r<R$ and $r>R$, the linear problem \eqref{e:Stab} can be written as
\begin{equation*}
0 = -(1 + \Delta_{k})^2 u_{k} + \mu u_{k}
\end{equation*}
for some $\mu\in\mathbb{R}$, which has a general solution of the form
\begin{equation*}
u_k = \tilde{A}_1 J_{k}(\nu_{+}\,r) + \tilde{A}_2 J_{k}(\nu_{-}\,r) + \tilde{A}_3 Y_{k}(\nu_{+}\,r) + \tilde{A}_4 Y_{k}(\nu_{-}\,r), \qquad\qquad \nu_\pm:= \sqrt{1 \pm\sqrt{\mu}},
\end{equation*}
where $J_k, Y_k$ are $k$'th order Bessel functions of the first and second kind. Hence, any real bounded solution $u_k$ of \eqref{e:Stab} takes the form
\begin{equation*}
u_k = \begin{cases}
A_1 J_{k}(\alpha_{+}\,r) + A_2 J_{k}(\alpha_{-}\,r), & r<R,\\
A_3 H^{(1)}_{k}( \beta r) +\overline{A_3}\, \overline{H^{(1)}_{k}( \beta r)}, & r>R,\\
\end{cases}  
\end{equation*}
with $A_1, A_2\in\mathbb{R}$ and $A_3\in\mathbb{C}$, where $\alpha_{\pm}=\sqrt{1\pm\varepsilon\sqrt{1-\tilde{\lambda}}}$, $\beta =\sqrt{1+\mathrm{i}\varepsilon\sqrt{1+\tilde{\lambda}}}$, and $H^{(1)}_{k}$ is the $k$'th order Hankel function, given by
\begin{equation*}
H^{(1)}_{k}(x) = J_{k}(x) + \mathrm{i}Y_{k}(x).
\end{equation*}
We decompose $\tilde{\lambda}$ and $\beta$ into their real and imaginary parts, which we denote as $\lambda_R, \lambda_I$ and $\beta_R, \beta_I$, respectively. 
The Hankel function $H_k^{(1)}$ then satisfies the following asymptotic behaviour
\begin{equation*}
    \begin{split}
        H^{(1)}_{k}(\beta\,r) ={}& \sqrt{\frac{2}{\pi\beta\,r}}\,\mathrm{e}^{\mathrm{i}(\beta\,r - \frac{(2k+1)\pi}{4})}\left(1 + \mathcal{O}(r^{-1})\right),\\
        ={}& \sqrt{\frac{2}{\pi\beta\,r}}\,\mathrm{e}^{-\beta_{I}\,r}\,\mathrm{e}^{\mathrm{i}\left(\beta_{R}\,r - \frac{(2k+1)\pi}{4}\right)}\left(1 + \mathcal{O}(r^{-1})\right),\\
    \end{split}
\end{equation*}
for $r\gg1$ and, using the formula $\sqrt{z} = \sqrt{|z|}\,\frac{z + |z|}{|z + |z||}$ for any $z\in\mathbb{C}$, we note that
\begin{equation*}
\begin{split}
    \beta_I ={}& \frac{\varepsilon\,\sqrt{(1 + \lambda_R)\, + \sqrt{(1 + \lambda_R)^2 + \lambda_I^2}}}{2\,\sqrt{1 -\frac{\varepsilon\,\lambda_I}{\sqrt{2}\,\sqrt{(1 + \lambda_R)\, + \sqrt{(1 + \lambda_R)^2 + \lambda_I^2}}} + \sqrt{1 + \varepsilon^2\,\sqrt{(1 + \lambda_R)^2 + \lambda_I^2} - \frac{\sqrt{2}\,\varepsilon\,\lambda_I}{\sqrt{(1 + \lambda_R)\, + \sqrt{(1 + \lambda_R)^2 + \lambda_I^2}}}}}}.
\end{split}
\end{equation*}
Hence, we observe that $\beta_I>0$ if and only if $1+\lambda_R + \sqrt{(1+\lambda_R)^2 + \lambda_I^2}>0$, and so $\tilde{\lambda} \notin(-\infty,-1]$ implies that bounded solutions $u_k$ to the linear problem \eqref{e:Stab} exhibit exponential radial decay and thus lie in $H^{4}_{(k)}(0,\infty)$. In order that a piecewise solution consisting of $u_{-}$ for $r<R$ and $u_{+}$ for $r>R$ is at least $C^3$ at $r=R$, the following matching conditions
\begin{equation*}
\begin{aligned}
u_{-}(R)={}&u_{+}(R), & \qquad
\mathcal{D}_{k}u_{-}(R)={}&\mathcal{D}_{k}u_{+}(R),\\
\Delta_{k} u_{-}(R)={}&\Delta_{k}u_{+}(R), &\qquad
\mathcal{D}_{k}\Delta_{k}u_{-}(R)={}&\mathcal{D}_{k}\Delta_{k}u_{+}(R),
\end{aligned}
\end{equation*}
must be satisfied. Hence, the constants $A_1, A_2, A_3$ must satisfy
\begin{equation*}
\begin{aligned}
A_1 u_1(R) + A_2 u_2(R) ={}& A_3 u_3(R) +\overline{A_3}\, \overline{u_3}(R), \\
A_1 \mathcal{D}_{k}u_1(R) + A_2 \mathcal{D}_{k}u_2(R) ={}& A_3 \mathcal{D}_{k}u_3(R) +\overline{A_3}\, \mathcal{D}_{k}\overline{u_3}(R), \\
A_1 \Delta_{k}u_1(R) + A_2 \Delta_{k}u_2(R) ={}& A_3 \Delta_{k}u_3(R) +\overline{A_3}\, \Delta_{k}\overline{u_3}(R), \\
A_1 \mathcal{D}_{k}\Delta_{k}u_1(R) + A_2 \mathcal{D}_{k}\Delta_{k}u_2(R) ={}& A_3 \mathcal{D}_{k}\Delta_{k}u_3(R) +\overline{A_3}\, \mathcal{D}_{k}\Delta_{k}\overline{u_3}(R), \\
\end{aligned}
\end{equation*}
where we have defined $u_1(r) = J_{k}(\alpha_{+}\,r)$, $u_2(r)= J_{k}(\alpha_{-}\,r)$ and $u_3(r) = H^{(1)}_{k}( \beta r)$. Since the Helmholtz equation $\Delta_k Z_k(\mu r) = - \mu^2 Z_{k}(\mu r)$ is satisfied for $Z_k \in\{J_k, Y_k, H_k^{(1)}\}$,
we obtain
\begin{equation*}
\begin{aligned}
A_1 u_1(R) + A_2 u_2(R) ={}& A_3 u_3(R) +\overline{A_3}\, \overline{u_3}(R), \\
A_1 \mathcal{D}_{k}u_1(R) + A_2 \mathcal{D}_{k}u_2(R) ={}& A_3 \mathcal{D}_{k}u_3(R) +\overline{A_3}\, \mathcal{D}_{k}\overline{u_3}(R), \\
A_1 u_1(R) -  A_2 u_2(R) ={}& \mathrm{i}\sqrt{\tfrac{1+\tilde{\lambda}}{1 - \tilde{\lambda}}} \, \left(A_3 u_3(R) - \overline{A_3}\, \overline{u_3}(R) \right), \\
 A_1 \mathcal{D}_{k}u_1(R) - A_2 \mathcal{D}_{k} u_2(R) ={}& \mathrm{i}\sqrt{\tfrac{1+\tilde{\lambda}}{1 - \tilde{\lambda}}} \, \left(A_3 \mathcal{D}_{k}u_3(R) - \overline{A_3}\, \mathcal{D}_{k}\overline{u_3}(R)\right), \\
\end{aligned}
\end{equation*}
which simplifies to
\begin{equation*}
\begin{aligned}
A_1 ={}& \frac{\rho\,\mathrm{e}^{\mathrm{i}\psi}\,A_3 u_3(R) + \rho\,\mathrm{e}^{-\mathrm{i}\psi}\, \overline{A_3}\, \overline{u_3}(R)}{u_1(R)}, \qquad\qquad
A_2 =\frac{\rho\,\mathrm{e}^{-\mathrm{i}\psi}\,A_3 u_3(R) + \rho\,\mathrm{e}^{\mathrm{i}\psi}\,\overline{A_3}\, \overline{u_3}(R)}{u_2(R)}, \\
\end{aligned}
\end{equation*}
and
\begin{equation*}
\begin{aligned}
0 ={}& \mathrm{e}^{\mathrm{i}\psi} A_3  W[u_1, u_3](R) + \mathrm{e}^{-\mathrm{i}\psi} \overline{A_3}\, W[u_1, \overline{u_3}](R), \\
0 ={}& \mathrm{e}^{-\mathrm{i}\psi} A_3 W[u_2, u_3](R) + \mathrm{e}^{\mathrm{i}\psi} \overline{A_3}\, W[u_2, \overline{u_3}](R),\\
\end{aligned}
\end{equation*}
where we recall that
\begin{equation*}
W[u, v](r) = r\left(u(r)\,v'(r) - v(r)\,u'(r)\right) = r\left(u(r)\,\mathcal{D}_{k}v(r) - v(r)\,\mathcal{D}_{k}u(r)\right),
\end{equation*}
and we have defined
\begin{equation*}
    \rho = \frac{1}{\sqrt{2}\,\sqrt{1-\tilde{\lambda}}},\qquad \qquad \mathrm{e}^{\mathrm{i}\psi} = \frac{1}{\sqrt{2}}\left(\sqrt{1-\tilde{\lambda}} + \mathrm{i}\,\sqrt{1+\tilde{\lambda}}\right).
\end{equation*}
Since the linear solution is defined up to a constant, we set $A_3 = \mathrm{e}^{\mathrm{i}\phi}\frac{1}{2\rho\,u_3(R)}$ and obtain
\begin{equation*}
u_k = \begin{cases}
\displaystyle \cos(\phi + \psi)\frac{J_{k}(\alpha_{+}\,r)}{J_{k}(\alpha_{+}\,R)} + \cos(\phi - \psi)\frac{J_{k}(\alpha_{-}\,r)}{J_{k}(\alpha_{-}\,R)}, & r<R,\\
\displaystyle \frac{1}{\rho}\, \mathrm{Re}\left(\mathrm{e}^{\mathrm{i}\phi}\, \frac{H^{(1)}_{k}( \beta r)}{H^{(1)}_{k}(\beta R)}\right), & r>R,\\
\end{cases}  
\end{equation*}
with $\phi$ given by
\begin{equation*}
\begin{aligned}
\mathrm{e}^{2\mathrm{i}\phi } ={}&-\mathrm{e}^{-2\mathrm{i}\psi}\frac{W[u_1, \overline{u_3}](R)u_3(R)}{W[u_1, u_3](R)\overline{u_3}(R)} = -\mathrm{e}^{2\mathrm{i}\psi}\frac{W[u_2, \overline{u_3}](R)u_3(R)}{W[u_2, u_3](R)\overline{u_3}(R)},\\
\end{aligned}
\end{equation*}
and where $R, \varepsilon, \tilde{\lambda}$ satisfy
\begin{equation*}
\begin{aligned}
F_{k}(R,\varepsilon,\tilde{\lambda}) ={}& \tfrac{1}{2\mathrm{i}} \left[\mathrm{e}^{\mathrm{i}2\psi} W[u_1, u_3](R)\,W[u_2, \overline{u_3}](R)- \mathrm{e}^{-\mathrm{i}2 \psi} W[u_1, \overline{u_3}](R)\,W[u_2, u_3](R) \right], \\
={}& \mathrm{Im}\left[\mathrm{e}^{\mathrm{i}2\psi}\, W[u_1, u_3](R)\,W[u_2, \overline{u_3}](R)\right]\\
={}& 0.
\end{aligned}
\end{equation*}
It follows from explicit calculations that 
\begin{equation*}
    \begin{split}
        \mathrm{e}^{\mathrm{i}2\psi} ={}& \frac{1}{2}\left(\sqrt{1-\tilde{\lambda}} + \mathrm{i}\,\sqrt{1+\tilde{\lambda}}\right)^2 =\mathrm{i}\left(\sqrt{1-\tilde{\lambda}^2} + \mathrm{i}\,\tilde{\lambda}\right) =\mathrm{i}\,\mathrm{e}^{\mathrm{i}\sin^{-1}(\tilde{\lambda})}
    \end{split}
\end{equation*}
and so
\begin{equation*}
\begin{aligned}
F_{k}(R,\varepsilon,\tilde{\lambda}) ={}& \mathrm{Re}\left[\mathrm{e}^{\mathrm{i}\sin^{-1}(\tilde{\lambda})}\, W[u_1, u_3](R)\,W[u_2, \overline{u_3}](R)\right].\\
\end{aligned}
\end{equation*}
\end{proof}

\subsection{Pre-compactness of solution branches}\label{app:comp}

In this section we attempt to prove that alternative $IV$ in Theorem~\ref{thm:Global} does not occur; that is, branches of solutions to \eqref{e:SH-steady} do not lose compactness for $\varepsilon>0$. For the remainder of this section we consider the following equation
\begin{equation}\label{e:1}
    \mathcal{L}_\varepsilon\,u = V_\varepsilon \, u + f_\varepsilon(u),
\end{equation}
where $\varepsilon>0$ is a fixed parameter, $u\in H^{2m}(\mathbb{R}^d)$ with $m>d/4$, and $f_\varepsilon$ is some real-analytic function that satisfies $f_\varepsilon(0)=f_\varepsilon'(0)=0$ and may depend on the parameter $\varepsilon$. The linear operator $\mathcal{L}_\varepsilon:H^{2m}(\mathbb{R}^d) \to L^2(\mathbb{R}^d)$ and potential function $V_\varepsilon = V_\varepsilon(\mathbf{x})$ are assumed to satisfy the following properties:
\begin{itemize}
    \item $\mathcal{L}_\varepsilon:H^{2m}(\mathbb{R}^d) \to L^2(\mathbb{R}^d)$ is defined as a Fourier multiplier, such that $\mathcal{F}[\mathcal{L}_\varepsilon u](\mathbf{k}) =\ell_\varepsilon(\mathbf{k})\,\mathcal{F}[u](\mathbf{k})$, where there exists some $d_\varepsilon>0$ such that the Fourier symbols $\ell_\varepsilon$ satisfies
    \begin{equation*}
        \ell_\varepsilon(\mathbf{k}) \geq d_\varepsilon \quad \text{for all $\mathbf{k}\in\mathbb{R}^d$,} \qquad\qquad|\ell_\varepsilon(\mathbf{k})|= \mathcal{O}(|\mathbf{k}|^{2m}) \quad \text{as $|\mathbf{k}|\to\infty$.}
    \end{equation*}
    \item $V_\varepsilon\in L^\infty(\mathbb{R}^d)$ and satisfies the following decay condition; for any $\delta>0$ there is some $T>0$ such that
    \begin{equation*}
        |V_\varepsilon(\mathbf{x})| \leq \delta \qquad \text{for all $|\mathbf{x}|>T$.}
    \end{equation*}
\end{itemize}
We emphasise that both the Swift-Hohenberg equation \eqref{e:SH} and the more general PDE \eqref{e:geneq} satisfy the above assumptions for $\mathcal{L}_\varepsilon$ and $V_\varepsilon$, and so we can apply any results for \eqref{e:1} directly to \eqref{e:SH-steady}.

We now prove that the global solution curve of an axisymmetric solution to \eqref{e:1} does not lose compactness. Note that in the following proof, with slight abuse of notation, we use $X^{2m}_{0}$ to denote the subspace of radially symmetric functions in $H^{2m}(\mathbb{R}^d)$.
\begin{lem}\label{lem:Compact-0}
    Let $\mathcal{C}\subset X^{2m}_{0}$ be a global bifurcation curve for a localised axisymmetric solution to \eqref{e:1}. Then, alternative IV in Theorem~\ref{thm:Global} cannot occur.
\end{lem}
\begin{proof}
Let $\{u_n\}$ be a bounded sequence in $X^{2m}_0$ of solutions to \eqref{e:1} and hence weakly convergent to some $u\in X^{2m}_0$. We note that $u_n\to u$ (up to subsequences) in $L^2(B_{T})$ for any $T>0$, where $B_{T}$ denotes the unit ball of radius $T$ about the origin. 

By Strauss' radial lemma~\cite[Lemma 1]{Strauss1977Radial}, there exists a constant $c_1>0$ such that
\begin{equation}\label{ineq:Strauss-0}
    |u(\mathbf{x})| \leq c_1\,|\mathbf{x}|^{-\frac{(d-1)}{2}}\,\| u\|_{H^1(\mathbb{R}^d)}
\end{equation}
holds for all $u\in X^1_0$. The condition that $f_\varepsilon$ satisfies $f_\varepsilon(0) = f_\varepsilon'(0)=0$ also implies that
\begin{equation*}
    |f_\varepsilon(u)| \leq c_2 |u|^{\alpha+1}
\end{equation*}
for all $u\in H^{2m}(\mathbb{R}^d)$ and some $\alpha>0$, where the constant $c_2>0$ is uniformly bounded in $H^{2m}(\mathbb{R}^d)$ (since $\|u\|_{L^\infty}$ is uniformly bounded). In particular, coupled with the uniform decay estimate \eqref{ineq:Strauss-0}, it follows that
\begin{equation*}
    |f_\varepsilon(u(\mathbf{x}))| \leq c_2 |u(\mathbf{x})|^{\alpha}\,|u(\mathbf{x})| \leq c_1^{\alpha}\,c_2 \,\| u\|_{H^1(\mathbb{R}^d)}^{\alpha} T^{-\frac{(d-1)\alpha}{2}} |u(\mathbf{x})|, \qquad\qquad \text{for all $|\mathbf{x}|\geq T$}
\end{equation*}
and so
\begin{equation*}
    \|f_\varepsilon(u)\|^2_{L^2(\mathbb{R}^d\backslash B_T)} \leq c_1^{2\alpha}\,c_2^2 \,\| u\|_{H^1(\mathbb{R}^d)}^{2\alpha} \|u\|^2_{L^2(\mathbb{R}^d)}\,  T^{-(d-1)\alpha}.
\end{equation*}
To prove that $u_n \to u$ strongly in $H^{2m}(\mathbb{R}^d)$, we note that $\mathcal{L}_\varepsilon:H^{2m}(\mathbb{R}^d) \to L^2(\mathbb{R}^d)$ is an isomorphism with
\begin{equation*}
    \| u\|_{H^{2m}(\mathbb{R}^d)} \leq \sup_{\mathbf{k}\in\mathbb{R}^d}\left(\frac{(1+|\mathbf{k}|^2)^m}{\ell_\varepsilon(\mathbf{k})}\right)\,\| \mathcal{L}_{\varepsilon}u\|_{L^2(\mathbb{R}^d)} =: C_\varepsilon\,\| \mathcal{L}_{\varepsilon}u\|_{L^2(\mathbb{R}^d)}
\end{equation*}
for all $u\in H^{2m}(\mathbb{R}^d)$. It thus follows that
\begin{equation*}
    \begin{split}
        \| u - u_n \|_{H^{2m}(\mathbb{R}^d)}^2 \leq{}& C_\varepsilon^2\,\|\mathcal{L}_\varepsilon (u - u_n) \|^2_{L^2(\mathbb{R}^d)},\\
        ={}& C_\varepsilon^2\,\left[\|\mathcal{L}_\varepsilon (u - u_n) \|^2_{L^2(B_T)} + \|\mathcal{L}_\varepsilon (u - u_n) \|^2_{L^2(\mathbb{R}^d\backslash B_T)}\right],\\
        \leq{}& C_\varepsilon^2\,\left[\|V_\varepsilon\|^2_{L^\infty}\,\| u - u_n\|_{L^2(B_T)}^2 + \|f_\varepsilon(u) - f_\varepsilon(u_n)\|^2_{L^2(B_T)}\right]\\
        &\qquad\qquad + 2C_\varepsilon^2\,\left[\|V_\varepsilon u\|^2_{L^2(\mathbb{R}^d\backslash B_T)} + \|V_\varepsilon  u_n\|^2_{L^2(\mathbb{R}^d\backslash B_T)}\right]\\
        &\qquad\qquad + 2C_\varepsilon^2\,\left[\|f_\varepsilon(u)\|^2_{L^2(\mathbb{R}^d\backslash B_T)} + \|f_\varepsilon(u_n)\|^2_{L^2(\mathbb{R}^d\backslash B_T)}\right],\\
        \leq{}& C_\varepsilon^2\,\left[\|V_\varepsilon\|^2_{L^\infty}\,\| u - u_n\|_{L^2(B_T)}^2 + \|f_\varepsilon(u) - f_\varepsilon(u_n)\|^2_{L^2(B_T)}\right]\\
        &\qquad\qquad + 2C_\varepsilon^2\,\left[\|u\|^2_{L^2(\mathbb{R}^d)} + \| u_n\|^2_{L^2(\mathbb{R}^d)}\right]\,\delta^2\\
        &\qquad\qquad + 2 c_1^{2\alpha}\,c_2^2 C_\varepsilon^2\,\left[\| u\|_{H^1(\mathbb{R}^d)}^{2\alpha} \|u\|^2_{L^2(\mathbb{R}^d)} + \| u_n\|_{H^1(\mathbb{R}^d)}^{2\alpha} \|u_n\|^2_{L^2(\mathbb{R}^d)}\right]\,T^{-(d-1)\alpha},\\
        \leq{}& \kappa,
    \end{split}
\end{equation*}
for any $\kappa>0$. The final step follows by choosing $T$ sufficiently large so that $T^{-(d-1)\alpha}$ and $\delta^2$ are both arbitrarily small, and then choosing $N = N(T)>0$ so that the first two terms are arbitrarily small for all $n>N$ (where we use that the each coefficient is uniformly bounded in $H^{2m}(\mathbb{R}^d)$, $u_n\to u$ strongly in $L^2(B_T)$ for any $T>0$, and $f_\varepsilon$ is at least continuous).
\end{proof}

The main difficulty in extending Lemma~\ref{lem:Compact-0} beyond axisymmetric solutions lies in reproducing the uniform decay estimate provided by Strauss' radial lemma. Categorising the uniform decay in general for $u\in H^1(\mathbb{R}^d)$ is significantly more difficult, as such functions are invariant under spatial translations and so we cannot expect that an estimate of the form \eqref{ineq:Strauss-0} holds for all $u\in H^1(\mathbb{R}^d)$. 

We present a slightly different formulation of Strauss' radial lemma for the special case $d=2$, which now holds for all $u\in H^1(\mathbb{R}^2)$.
\begin{lem}\label{lem:Strauss}
    Fix $k\in\mathbb{Z}$, suppose $u \in H^1(\mathbb{R}^2)$, and define $u_k\in H^1(\mathbb{R}^2)$ to be the mode-$k$ component of $u$, given by
    \begin{equation*}
        u_k(\mathbf{x}) = \frac{1}{2\pi}\int_{0}^{2\pi} u(\mathcal{R}_\theta\mathbf{x})\,\mathrm{e}^{-\mathrm{i}k\theta}\,\mathrm{d}\theta,
    \end{equation*}
    where $\mathcal{R}_\theta$ denotes the standard rotation matrix in $\mathbb{R}^2$ about an angle $\theta\in[0,2\pi)$. Then, there exists a constant $C>0$ such that the inequalities
    \begin{equation*}
    \begin{split}
        |u_k(\mathbf{x})| \leq{}& (2\pi)^{-\frac{1}{2}}\,C\,|\mathbf{x}|^{-\frac{1}{2}}\,\| u_k\|_{H^{1}(\mathbb{R}^2)},\qquad \qquad  \int_{0}^{2\pi} |u(\mathcal{R}_\theta\mathbf{x})|^2\,\mathrm{d}\theta \leq C^2\,|\mathbf{x}|^{-1}\,\| u\|^2_{H^{1}(\mathbb{R}^2)},
    \end{split}
    \end{equation*}
    hold for all $\mathbf{x}\in\mathbb{R}^2$. In particular, the constant $C$ is given by $C = \left(\max_{s\geq0} \{s\,I_0(s)\,K_0(s)\}\right)^{\frac{1}{2}}$.
\end{lem}
\begin{proof}
We begin by noting that any $u\in H^1(\mathbb{R}^2)$ can be expressed via the modal decomposition
\begin{equation}\label{id:u-fourier}
    u = \sum_{k\in\mathbb{Z}} u_k, 
\end{equation}
where the infinite sum converges in $H^1(\mathbb{R}^2)$ such that we can write
\begin{equation*}
    \sum_{k\in\mathbb{Z}} \| u_k \|_{H^1(\mathbb{R}^2)}^2 = \| u\|_{H^1(\mathbb{R}^2)}^2.
\end{equation*}
We note that the \textit{mode-$k$ function} $u_k$ satisfies the following identities
\begin{equation*}
    u_k(r\cos\theta,r\sin\theta) = \mathrm{e}^{\mathrm{i}k\theta}\hat{u}_k(r), \qquad\qquad \mathcal{F}[u_k](\rho\cos\omega,\rho\sin\omega) = \mathrm{i}^{-k}\mathrm{e}^{\mathrm{i}k\omega}\mathcal{H}_{k}[\hat{u}_k](\rho), 
\end{equation*}
where $\hat{u}_k$ is termed the \textit{radial coefficient} of $u_k$ (cf.\ Groves and Hill \cite{groves2024function}), and $\mathcal{H}_{k}$ denotes the $k$-index Hankel transform given by
\begin{equation*}
    \mathcal{H}_k[\hat{u}_k](\rho) = \int_{0}^{\infty} J_{k}(\rho\,r)\,\hat{u}_k(r)\,r\,\mathrm{d}r.
\end{equation*}
The Hankel transform is self-inverse, and so we obtain the following estimate
\begin{equation*}
    \begin{split}
        |u_k(\mathbf{x})| ={}& |\hat{u}_k(|\mathbf{x}|)|,\\
        ={}& \left|\int_{0}^{\infty} J_{k}(|\mathbf{x}|\,\rho)\,\mathcal{H}_{k}[\hat{u}_k](\rho)\,\rho\,\mathrm{d}\rho\right|,\\
        \leq{}& \left(\frac{1}{2\pi}\int_{0}^{\infty} \frac{|J_{k}(|\mathbf{x}|\,\rho)|^2}{(1+\rho^2)}\,\rho\,\mathrm{d}\rho\right)^{\frac{1}{2}}\left(2\pi\int_{0}^{\infty} (1+\rho^2)\,|\mathcal{H}_{k}[\hat{u}_k](\rho)|^2\,\rho\,\mathrm{d}\rho\right)^{\frac{1}{2}},\\
        ={}& \left(\frac{1}{2\pi} I_{k}(|\mathbf{x}|)\,K_{k}(|\mathbf{x}|)\right)^{\frac{1}{2}}\left(\int_{\mathbb{R}^2}(1+|\boldsymbol{\xi}|^2)\,|\mathcal{F}[u_k](\boldsymbol{\xi})|^2\,\mathrm{d}\boldsymbol{\xi}\right)^{\frac{1}{2}},\\
        \leq{}& (2\pi)^{-\frac{1}{2}}\left(\max_{s\geq0} \{s\,I_{0}(s)\,K_{0}(s)\}\right)^{\frac{1}{2}}\,|\mathbf{x}|^{-\frac{1}{2}}\,\| u_k\|_{H^1(\mathbb{R}^2)}.\\
    \end{split}
\end{equation*}
The last step follows from the fact that the function $r\mapsto r I_k(r)\,K_k(r)$ converges to $\frac{1}{2}$ as $r\to\infty$ for all $k\in\mathbb{Z}$, is monotone increasing for all $k \in \mathbb{Z}\backslash\{0\}$, and is continuous and converges to $0$ as $r\to0$ for $k=0$. Hence we have derived the Strauss-type decay estimate
\begin{equation}\label{id:Strauss;k}
    \begin{split}
        |u_k(\mathbf{x})|  \leq{}& (2\pi)^{-\frac{1}{2}}\,C\,|\mathbf{x}|^{-\frac{1}{2}}\,\| u_k\|_{H^1(\mathbb{R}^2)}.\\
    \end{split}
\end{equation}
We note that the modal decomposition \eqref{id:u-fourier} yields the following Fourier expansion for $u(\mathcal{R}_\theta\,\cdot)$ 
\begin{equation*}
    u(\mathcal{R}_\theta\,\cdot) = \sum_{k\in\mathbb{Z}} \mathrm{e}^{\mathrm{i}k\theta}\,u_k,
\end{equation*}
and so applying Parseval's identity yields
\begin{equation*}
    \frac{1}{2\pi}\int_{0}^{2\pi} |u(\mathcal{R}_\theta\mathbf{x})|^2\,\mathrm{d}\theta = \sum_{k\in\mathbb{Z}} |u_k(\mathbf{x})|^2.
\end{equation*}
Applying the Strauss-type estimate \eqref{id:Strauss;k} for each $k\in\mathbb{Z}$, we obtain
\begin{equation*}
    \int_{0}^{2\pi} |u(\mathcal{R}_\theta\mathbf{x})|^2\,\mathrm{d}\theta \leq \sum_{k\in\mathbb{Z}} C^2\,|\mathbf{x}|^{-1}\|u_k\|^2_{H^1(\mathbb{R}^2)} = C^2\,|\mathbf{x}|^{-1}\|u\|^2_{H^1(\mathbb{R}^2)}
\end{equation*}
where the last step follows from the $H^1$-convergence of \eqref{id:u-fourier} (which implies that the infinite sum in Parseval's identity converges uniformly, due to the Weierstrass M-test).
\end{proof}

While this lemma provides an equivalent decay rate to that found by Strauss in the axisymmetric case, the result does not provide a uniform decay estimate for $u\in H^1(\mathbb{R}^2)$ but rather on the spherical mean of $u$. This avoids the problem of translational invariance in $H^1(\mathbb{R}^2)$, but means that we are unable to simply apply the same argument as in the proof of Lemma~\ref{lem:Compact-0}. In particular, we require an estimate of the form
\begin{equation*}
    \int_{0}^{2\pi} |u(\mathcal{R}_\theta\mathbf{x})|^{2(1+\alpha)}\,\mathrm{d}\theta \leq c_1^2\,\rho(T)\,\int_{0}^{2\pi} |u(\mathcal{R}_\theta\mathbf{x})|^2\,\mathrm{d}\theta \qquad \text{for all $|\mathbf{x}|>T$,}
\end{equation*}
where $\rho$ is some function that satisfies $\rho(T)\to0$ as $T\to\infty$. If such an estimate can be obtained, then the pre-compactness of the global branches would follow identically to the axisymmetric case presented in Lemma~\ref{lem:Compact-0}.

\begin{rmk}\label{rmk:Compactness-extend}
    We expect that all bifurcating solutions $u\in H^{2m}(\mathbb{R}^d)$ to \eqref{e:1} decay at an exponential rate as $|\mathbf{x}|\to\infty$. One can characterise this decay property via radial exponential dichotomies (as recently developed by Beck, Goh \& Haslam-Hyde~\cite{beck2025}) or by studying the associated Green's function of $\mathcal{L}_\varepsilon$ (where one can use Paley--Wiener theory to characterise the decay of a general Fourier symbol $\ell_\varepsilon$, such as in the work of Arnesen~\cite{Arnesen2022Decay} for one-dimensional Green's functions). However, in both of these approaches it remains difficult to obtain the \emph{uniform} exponential decay rates required in order to prove compactness. 
\end{rmk}

\bibliographystyle{abbrv}
\bibliography{Bibliography.bib}

\end{document}